\documentclass[11pt]{article}%
\usepackage{amssymb}
\usepackage{amsfonts}
\usepackage{graphicx}
\usepackage{amsmath}
\usepackage{comment}
\usepackage{amsthm}
\usepackage{caption}
\usepackage{booktabs}
\usepackage{multirow}
\usepackage[titletoc]{appendix}
\usepackage{float}
\usepackage[table]{xcolor}
\usepackage{indentfirst}
\usepackage[nohead]{geometry}
\usepackage[singlespacing]{setspace}
\usepackage[bottom]{footmisc}
\usepackage{indentfirst}
\usepackage[round]{natbib}
\usepackage{subfig}%
\usepackage{hyperref}
\hypersetup{colorlinks,allcolors=black}
\usepackage{enumitem}

\usepackage{bbm}

\providecommand{\U}[1]{\protect \rule{.1in}{.1in}}
\graphicspath{{C:/Users/jiliu/Desktop/plots/type1}}
\newtheorem{theorem}{Theorem}[section]

\newtheorem{assumption}{Assumption}[section]

\newtheorem{corollary}{Corollary}[section]

\newtheorem{lemma}{Lemma}[section]

\newtheorem{proposition}{Proposition}[section]

\numberwithin{equation}{section}

\theoremstyle{definition}
\newtheorem{definition}{Definition}[section]
\newtheorem{example}{Example}[section]
\newtheorem{remark}{Remark}[section]

\newcommand{\be}{\begin{eqnarray}}
	\newcommand{\ee}{\end{eqnarray}}
\newcommand{\by}{\begin{eqnarray*}}
	\newcommand{\ey}{\end{eqnarray*}}
\newcommand{\bn}{\begin{enumerate}}
	\newcommand{\en}{\end{enumerate}}
\newcommand{\bi}{\begin{itemize}}
	\newcommand{\ei}{\end{itemize}}

\newcommand{\one}{\mathbf{1}}

\renewcommand{\geq}{\geqslant}
\renewcommand{\leq}{\leqslant}
\renewcommand{\epsilon}{\varepsilon}

\makeatletter
\def \@biblabel#1{\hspace*{-\labelsep}}
\makeatother

\begin{document}
\title{Asymptotic Properties of Generalized Shortfall Risk Measures for Heavy-tailed Risks}

\author{ Tiantian Mao$^{[a]}$~ \thinspace \
	Gilles Stupfler$^{[b]}$~ \thinspace \
	Fan Yang$^{[c]}$ \\
	%EndAName
	$[a]$ {\small Department of Statistics and Finance, School of Management,
		School of Data Science}\\
	{\small \ University of Science and Technology of China, Hefei 230026, P. R.
		China}\\
	$[b]$ {\small Univ Angers, CNRS, LAREMA, SFR MATHSTIC, F-49000 Angers, France}\\
	$[c]${\small \ Department of Statistics and Actuarial Science, University of
		Waterloo }\\
	{\small Waterloo, ON N2L 3G1, Canada }}
\date{{\small May 4, 2023}}
\maketitle

\begin{abstract}
We study a general risk measure called the generalized shortfall risk measure,
which was first introduced in \cite{mao2018risk}. It is proposed under the
rank-dependent expected utility framework, or equivalently induced from the cumulative
prospect theory. This risk measure can be flexibly designed to capture the
decision maker's {behavior} toward risks and wealth when measuring risk. In this paper, we derive the first- and second-order asymptotic
expansions for the generalized shortfall risk measure. Our asymptotic results can be viewed as unifying theory for, among others, distortion risk measures and utility-based shortfall risk measures.
%Since it contains a large class of risk measures as special cases including distortion risk measures, utility-based shortfall risk measures and etc, our asymptotic results can be viewed as a unified results.  
 {They also provide a blueprint for the estimation of these measures at extreme levels, and we illustrate this principle by constructing and studying a quantile-based estimator in a special case. The accuracy of the asymptotic expansions and of the estimator is assessed on several numerical examples.}

	\bigskip
		
		\textbf{Keywords}: Generalized shortfall risk measure, Asymptotic expansions, Heavy tails,  {Estimation}
\end{abstract}

\section{Introduction}

In this paper we study extreme value properties of a general risk measure, called the generalized shortfall risk measure and defined as follows. Let $u_{1}$, $u_{2}$ be
{(strictly)} increasing functions on $\mathbb{R}_+$ with $u_{1}(0)=u_{2}(0)=0$, and $h_{1}$
and $h_{2}$ be two distortion functions on {$[0,1]$, supposed to be right-continuous and increasing throughout,}
%satisfying that they are increasing and 
such that $h_{i}(0)=0$ and $h_{i}(1)=1$ with no jumps at $0$ and $1$. For
a random variable $X$ with distribution function $F$, the generalized
shortfall risk measure, denoted by $x_{\tau}=x_{\tau}(X,u_1,h_1,u_2,h_2)$, is defined as the solution to the following equation:
\begin{align}
\label{eq-eqREDU-2}
    \tau \, \mathrm{H}_{u_{1},\,h_{1}}((X-x)_{+})&=(1-\tau)\, \mathrm{H}%
	_{u_{2},\,h_{2}}((X-x)_{-}), \\ % 
\nonumber
	\mbox{where } \mathrm{H}_{u_{1},\,h_{1}}((X-x)_{+})&=\int_{x}^{\infty}u_{1}(y-x)\mathrm{\,d}%
h_{1}(F(y)), \\
\nonumber
    \mbox{and } \mathrm{H}_{u_{2},\,h_{2}}((X-x)_{-})&=\int_{-\infty}^{x}u_{2}(x-y)\mathrm{\,d}%
h_{2}(F(y)).
\end{align}
This problem is written under the appropriate regularity and integrability assumptions making both sides in~\eqref{eq-eqREDU-2} finite and ensuring that the solution is indeed unique; see Section~\ref{first-order} below for a discussion.
The generalized shortfall risk measure was %in fact 
first introduced in
\cite{mao2018risk} %by generalizing 
as an extension of the generalized quantile risk measure. The
quantile of random variable $X$, or generalized (left-continuous) inverse function at level
$\tau \in(0,1)$, or Value-at-Risk (VaR), is defined as $F^{\leftarrow}(\tau)=\inf \{x\in \mathbb{R}, F(x)\geq \tau \}$.
It is well known that $F^{\leftarrow}(\tau)$ can also be represented as%
\[
F^{\leftarrow}(\tau)=\underset{x\in \mathbb{R}}{\arg \min}\left \{
\tau \mathbb{E}[(X-x)_{+}]+(1-\tau)\mathbb{E}[(X-x)_{-}]\right \}  ,
\]
where $x_{+}=\max \{x,0\}$ and $x_{-}=\max \{-x,0\}=-\min \{x,0\}$, provided $\mathbb{E}|X|<\infty$. By transforming the
shortfall risk $(X-x)_{+}$ to $\phi_{1}(  (X-x)_{+})  $ and the (in the terminology of~\citealp{mao2018risk}) over-required capital risk $(X-x)_{-}$ to $\phi_{2}(  (X-x)_{-})
$, where $\phi_{1}$, $\phi_{2}$ are increasing convex functions, the
generalized quantile was proposed in \cite{bellini2014generalized} as%
\begin{equation}
	\underset{x\in \mathbb{R}}{\arg \min}\left \{  \tau \mathbb{E}[\phi_{1}(
	(X-x)_{+})  ]+(1-\tau)\mathbb{E}[\phi_{2}(  (X-x)_{-})
	]\right \}  .\label{GI}%
\end{equation}
When $\phi_{1}(x)=\phi_{2}(x)=x^{2}$, the generalized quantile (\ref{GI})
reduces to %a well-known risk measure, the 
the well-known expectile, which was proposed in
\cite{newey1987asymmetric}. Since both the shortfall risk and over-required capital risk are evaluated under the original probability measure, the
generalized quantile is defined in the sense of the classical expected
utility. \cite{mao2018risk} further generalized it by using the rank-dependent
expected utility (RDEU) to evaluate risks and wealth (see e.g.
\citealp{alma9911678653505162}). To be more specific, letting $\phi_{1}$ and
$\phi_{2}$ be two nondegenerate increasing convex functions on %$\mathbb{R}%
%^{+}$, 
$[0,\infty)$, the so-called generalized quantile based on RDEU theory is defined as
\begin{equation}
	\underset{x\in \mathbb{R}}{\arg \min}\left \{  \tau \mathrm{H}_{\phi_{1},h_{1}%
	}((X-x)_{+})+(1-\tau)\mathrm{H}_{\phi_{2},h_{2}}((X-x)_{-})\right \}
	.\label{rdeu}%
\end{equation}
In \cite{mao2018risk}, Proposition 2.2 (ii) shows that if $u_{1}(x)=\phi
_{1}^{\prime}(x)$ and $u_{2}(x)=\phi_{2}^{\prime}(x)$, then the generalized
quantile based on RDEU theory defined in (\ref{rdeu}) coincides with the
generalized shortfall risk measure in (\ref{eq-eqREDU-2}). This shows that
besides the utility functions {that} can be selected, the generalized shortfall risk
measure allows decision makers to choose the appropriate distorted
probability measure to describe their  {behavior} towards risks and
wealth. This brings great flexibility in measuring the risk. 

Further, Theorem 3.1 of \cite{mao2018risk} showed that the generalized shortfall risk measure
is equivalent to the so-called generalized shortfall induced by cumulative prospect theory
(CPT) defined in \eqref{cpt} below when $v$, $h_1$ and $h_2$ are chosen properly.  CPT was
proposed by \cite{TverskyAmos1992AiPT} and has been applied in various areas
such as portfolio selection and pricing insurance contracts; see {\it e.g.}
\cite{schmidt2007linear}, \cite{kaluszka2012mean}, 
\cite{kaluszka2012pricing} and \cite{jin2013greed}. For an increasing continuous function $v$ on
$\mathbb{R}$, the generalized shortfall induced by CPT is defined as
\begin{equation}
	\inf \{x\in \mathbb{R},H_{v,h_{1},h_{2}}(X-x)\leq0\},\label{cpt}%
\end{equation}
where
\[
H_{v,h_{1},h_{2}}(X)=\int_{-\infty}^{0}v(y)\mathrm{\,d}h_{1}(F(y))+\int
_{0}^{\infty}v(y)\mathrm{\,d}h_{2}(F(y)).
\]
The generalized shortfall risk measure understood in the form of
(\ref{cpt}) contains utility-based shortfall risk measures
(\citealp{FollmerSchied+2016}) as special cases.

 {%After the 2008 Global Financial Crisis, much  research has been focusing on the study of tail risks and their disastrous consequences. In fact, many empirical studies have shown that asset returns in finance and large losses in insurance exhibit heavy tails; see, for example, \cite{loretan1994testing}, \cite{gabaix2003theory}, and \cite{gabaix2009power}. 
The study of tail risks and their disastrous consequences in finance has attracted substantial attention and many empirical studies have shown that asset returns in finance and large losses in insurance exhibit heavy tails: see, for example, \cite{loretan1994testing}, \cite{gabaix2003theory}, and \cite{gabaix2009power}. Moreover, regulators such as Basel III have recommended to estimate VaR with a confidence level very close to $1$. In the same spirit, we are interested in the behavior of the generalized shortfall risk measure for heavy-tailed risks when the confidence level $\tau$ is close to $1$. However, since the generalized shortfall risk measure extends in particular the expectile, for which no closed form is available in general, no simple explicit expression of $x_{\tau}$ is available, which makes the study of the risk measure for heavy-tailed risks difficult. Asymptotic expansions of risk measures, in terms of the quantile of the random variable of interest (viewed as a well-understood risk measure), provide an intuitive way to study extreme risk measures for heavy-tailed risks; see for example, the asymptotic expansions of the Haezendonck--Goovaerts risk measure in \cite{tang2012haezendonck} and \cite{mao2012second}, the conditional tail expectation in \cite{hua2011second} and \cite{hua2014strength}, the expectiles in \cite{bellini2014generalized} and \cite{mao2015asymptotic}, the risk concentration based on expectiles in \cite{mao2015risk}.} 

 {It is precisely the objective of this paper to study the first- and second-order asymptotic expansions of $x_{\tau}$ for a heavy-tailed random variable $X$ as the confidence level $\tau$ converges to $1$. From the technical point of view, the methodology used in this paper to derive the expansions is very general, in the sense that it can be applied to derive the asymptotic expansions of other quantile-based risk measures. A potential statistical benefit of such results is that, while the lack of a simple explicit expression of $x_{\tau}$ makes the estimation and practical use of $x_{\tau}$ difficult, an asymptotic expansion in terms of extreme quantiles is helpful in studying the asymptotic behavior of simple plug-in estimators of $x_{\tau}$ at extreme levels (see the so-called indirect estimator of~\citealp{daouia2018estimation}). Such expansions also allow to quantify bias terms and are fundamental in the derivation of asymptotic normality results for estimators at extreme levels. This estimation approach of an extreme risk measure has been adopted for, among others, the estimation of the marginal expected shortfall in \cite{cai2015estimation}, expectiles in \cite{daouia2018estimation}, M-quantiles in \cite{daouia2019Lpestimation} and the Haezendonck–Goovaerts risk measure in \cite{zhao2021estimation}. We illustrate that in this article by studying the asymptotic properties of an estimator of $x_{\tau}$ (with $\tau=\tau_n\uparrow 1$ as the size $n$ of the available sample of data tends to infinity) based on extreme quantiles of a distorted version of the underlying distribution. Our high-level result may be valid even when serial dependence is present in the data, as long as the observations come from a strictly stationary sequence.}

 {The rest of the paper is organized as follows. Section \ref{prl} %introduces the preliminaries on the regular variations. 
provides necessary technical background on regular variation. In Sections \ref{first-order} and \ref{second-order}, we derive the
first- and second-order expansions of the generalized shortfall risk measure,
respectively. Section~\ref{sec:estim} discusses the estimation of the generalized shortfall risk measure at extreme levels. In Section \ref{example}, we %present a numerical example 
give a couple of examples where our theory applies and we discuss a small-scale simulation study illustrating the performance of our estimator. All the proofs are relegated to Section \ref{proof}.}

%\section{Preliminary}
\section{Regular variation}\label{prl}

We start by introducing regular variation conditions that will be the backbone of our model on risk variables. 

% \begin{definition}
% 	An eventually nonnegative measurable function $f(\cdot)$ is said to be
% 	\emph{regularly varying} at $t_{0}\in\overline{\mathbb{R}}:=\mathbb{R}\cup\{
% 	\pm\infty\}$ with regularity index $\alpha\in\mathbb{R}$, if for all
% 	$x>0$,
% 	\begin{equation}
% 	\lim_{t\rightarrow t_{0}}\frac{f(tx)}{f(t)}=x^{\alpha}.
% 	\label{regular varying}%
% 	\end{equation}
% 	Denote this by $f(\cdot)\in\mathrm{RV}_{\alpha}(t_{0})$. When $t_{0}=+\infty$,
% 	write $f(\cdot)\in\mathrm{RV}_{\alpha}$.
% \end{definition}
%
\begin{definition}
An eventually nonnegative measurable function $f(\cdot)$ is said to be
	\emph{regularly varying} at $\infty$ with index $\alpha\in\mathbb{R}$, if for all
	$x>0$,
	\begin{equation}
	\lim_{t\rightarrow \infty}\frac{f(tx)}{f(t)}=x^{\alpha}.
	\label{regular varying}%
	\end{equation}
	We write $f(\cdot)\in\mathrm{RV}_{\alpha}$.
\end{definition}
%
%Especially for 
For a random variable $X$ with distribution function $F$, we %call $X$ is a 
say that $X$ is regularly varying with extreme value index $\gamma>0$ %index $\alpha<0$ 
if its survival function $\overline{F}=1-F$ is regularly varying with index %$\alpha$. This is also denoted by $X\in{\rm RV}_{\alpha}$, $\alpha<0$.
$-1/\gamma$. This is also denoted by $X\in{\rm RV}_{-1/\gamma}$. %$\alpha<0$.

 For a distribution function $F$, its (left-continuous inverse) quantile function is defined as
$F^{\leftarrow}(p)=\inf\{x\in\mathbb{R}:F(x)\geq p\}$ for $p\in(0,1)$. The
\emph{tail quantile function} $U(\cdot)$ of $F$ is defined as $$U(t
)=\left(\frac1{\overline{F}}\right)^{\leftarrow}(t)=F^{\leftarrow}
\left(1-\frac1t\right),~~{t>1}.$$ The RV
definition of a survival function $\overline{F}=1-F$ can be equivalently presented in terms of the tail quantile function $U$ %through the following definition of extended regular variation (ERV); see Section B.2 of \cite{de2006extreme}.
by requiring that $U\in\mathrm{RV}_{\gamma}$, that is, 
\[
\forall x>0, \ \lim_{t\rightarrow \infty}\frac{U(tx)}{U(t)}=x^{\gamma}.
\]
This assumption connects tail quantiles to arbitrarily extreme quantiles further away in the right tail through the approximation $F^{\leftarrow}(p')\approx \lbrack(1-p')/(1-p)]^{-\gamma}F^{\leftarrow}(p)$, for
$0<p<p'<1$ both close to 1.%

In practice it is necessary to quantify the bias incurred through the use of this approximation. This is typically done thanks to a second-order regular variation condition, itself most easily written using the concept of extended regular variation, which we recall below.

% \begin{definition}
% 	An eventually nonnegative measurable function $f(\cdot)$ is said to be
% 	\emph{extended regularly varying} at $\infty$, with an index $\gamma
% 	\in\mathbb{R}$ and auxiliary function $a(\cdot)>0$, if for all $x>0$,
% 	\begin{equation}
% 	\lim_{t\rightarrow\infty}\frac{f(tx)-f(t)}{a(t)}=\frac{x^{\gamma}-1}{\gamma}.
% 	\label{ERV}%
% 	\end{equation}
% When $\gamma=0$, the limit $(x^{\gamma}-1)/{\gamma}$ is understood as $\log x$.	Denote this by $f(\cdot)\in\mathrm{ERV}_{\gamma}$.
% \end{definition}

\begin{definition}
	A measurable function $f(\cdot)$ on $(0,\infty)$ is said to be
	\emph{extended regularly varying} at $\infty$, with an index $\gamma
	\in\mathbb{R}$ and an auxiliary function $a(\cdot)$ having constant sign, if for all $x>0$,
	\begin{equation}
	\lim_{t\rightarrow\infty}\frac{f(tx)-f(t)}{a(t)}=\frac{x^{\gamma}-1}{\gamma}.
	\label{ERV}%
	\end{equation}
When $\gamma=0$, the limit $(x^{\gamma}-1)/{\gamma}$ is understood as $\log x$.	Denote this by $f(\cdot)\in\mathrm{ERV}_{\gamma}$.
\end{definition}

Compared to the definition of extended regular variation in Section~B.2 of~\cite{de2006extreme}, we absorb the potential multiplicative constant appearing in the right-hand side into the function $a$. This results in a simpler limit, but the auxiliary function is allowed to be negative.

This allows us to introduce second-order regular variation as a special case of extended regular variation through the following definition.

%It is easy to check that for a distribution $F$, its survival function $\overline{F}\in {\rm RV}_{-\alpha}$ for $\alpha>0$ if and only if $U\in {\rm ERV}_{1/\alpha}$. Next, we introduce the second-order condition.

% \begin{definition}
% 	A regularly varying function $f(\cdot)$ is said to be \emph{second-order
% 		regularly varying} at $t_{0}\in\overline{\mathbb{R}}$ with first-order index
% 	$\gamma\in\mathbb{R}$ and second-order index $\rho\leq0$, if there exists a
% 	measurable function $A(\cdot)$, which does not change sign eventually and
% 	converges to $0$, such that, for all $x>0$,
% 	\begin{equation}
% 	\lim_{t\rightarrow t_{0}}\frac{f(tx)/f(t)-x^{\gamma}}{A(t)}=x^{\gamma}%
% 	\frac{x^{\rho}-1}{\rho}=:J_{\gamma,\rho}(x). \label{eq-d2}%
% 	\end{equation}
% 	When $\rho=0$, $J_{\gamma,\rho}(x)$ is understood as $x^{-\gamma}\log x$.
% 	Denote this by $f(\cdot)\in2\mathrm{RV}_{\gamma,\rho}(t_{0})$ and $A$ is
% 	called an auxiliary function. When $t_{0}=+\infty$, write $f(\cdot
% 	)\in\mathrm{2RV}_{\gamma,\rho}$.
% \end{definition}

\begin{definition}\label{2rvdef}
	A regularly varying function $f(\cdot)$ is said to be \emph{second-order
		regularly varying} at $\infty$ with first-order index
	$\gamma\in\mathbb{R}$ and second-order index $\rho\leq 0$, if there exists a
	measurable function $A(\cdot)$, which does not change sign eventually and
	converges to $0$, such that $t\mapsto t^{-\gamma} f(t)\in \mathrm{ERV}_{\rho}$ with auxiliary function $a:t\mapsto t^{-\gamma} f(t) A(t)$. In other words, 
	\begin{equation}
	\lim_{t\to \infty}\frac{f(tx)/f(t)-x^{\gamma}}{A(t)}=x^{\gamma}%
	\frac{x^{\rho}-1}{\rho}=:J_{\gamma,\rho}(x). \label{eq-d2}%
	\end{equation}
	When $\rho=0$, $J_{\gamma,\rho}(x)$ is understood as $x^{\gamma}\log x$.
	We write $f(\cdot
	)\in\mathrm{2RV}_{\gamma,\rho}$ and $A$ is called (second-order) auxiliary function.
\end{definition}

It is worth noting that each of the convergences in \eqref{regular varying},
(\ref{ERV}) and \eqref{eq-d2} is uniform with respect to $x$ in any compact
subset of {$(0,\infty)$:} %$\mathbb{R}_{+}$; 
see, for example, Theorem B.1.4 and Theorem B.2.9
of \cite{de2006extreme}. In our specific context of heavy-tailed distributions, Theorem 2.3.9 of \cite{de2006extreme} shows that for $\gamma>0$ and $\rho \leq0$,
$U(\cdot)\in2\mathrm{RV}_{\gamma,\rho}$ with an auxiliary function $A(\cdot)$
if and only if $\overline{F}(\cdot)\in2\mathrm{RV}_{-1/\gamma,\rho/\gamma}$
with an auxiliary function $A(1/\overline{F}(\cdot))$. In this case necessarily $A(\cdot)\in\mathrm{RV}_{\rho}$.

Lastly, we present two useful expansions of a 2RV function and its inverse
function, when the second-order parameter is negative; see {the proof of Lemma 3 in \cite{hua2011second} for Lemma \ref{lm-2RV-rep-inv} (i) and Proposition 2.5 of \cite{mao2012second} for Lemma \ref{lm-2RV-rep-inv} (ii).}
{
\begin{lemma} \label{lm-2RV-rep-inv} Let %$\gamma \in \Re, \rho<0$ 
$\gamma\in \mathbb{R}$, $\rho<0$ and $A$ be a measurable function having constant sign. 
\begin{enumerate}[label=(\roman*)] 
\item Then $h \in 2 \mathrm{RV}_{\gamma, \rho}$ with auxiliary function $A(\cdot)$ if and only if there exists a constant $c>0$ such that
\[
h(t)=c t^{\gamma}\left[1+\frac{1}{\rho} A(t)+o(A(t))\right], \ t \rightarrow \infty.
\]
\end{enumerate}

\begin{enumerate}[label=(\roman*),resume] 
\item {Then, when $\gamma>0$ and with the notation of $(i)$,} $h^{\leftarrow}$ has the following representation:
\[
h^{\leftarrow}(t)=c^{-1 / \gamma} t^{1 / \gamma}\left[1-\frac{1}{\gamma \rho} A(h^{\leftarrow}(t))+o(A(h^{\leftarrow}(t)))\right], \ t \rightarrow \infty.
\]
 {In particular $h^{\leftarrow} \in 2 \mathrm{RV}_{1/\gamma, \rho/\gamma}$.}
%\item For any $d \neq 0$, $h(t+d A(t))=h(t)(1+o(A(t)))$ as $t \rightarrow \infty$.
\end{enumerate}
\end{lemma}
}
\section{First-order expansions of  generalized shortfall risk measures} \label{first-order}

In this section, we study the first-order asymptotics of generalized shortfall risk measures. All the proofs are relegated to Section \ref{proof}.

%Let $Y$ be a random variable such that its distribution function is given by $h(F)$. Then with $u_1(0)=0$, we have
%\begin{align*}
%{\rm H}_{u_1, \, h_1} ((X-x)_+)& = \int_x^\infty u_1(y-x) \d h_1(F(y))\\
%& = \E[u_1(Y-x)_+] \\
%&  = \int_0^\infty \p(u_1((Y-x)_+) > t) \d t\\
%& = \int_0^\infty \p((Y-x)_+  > u_1^{-1}(t)) \d t\\
%& = \int_0^\infty \p( Y-x   > u_1^{-1}(t)) \d t\\
%& = \int_0^\infty \p( Y   > u_1^{-1}(t)+x) \d t\\
%& = \int_0^\infty 1-h_1(F(u_1^{-1}(t)+x))) \d t\\
%& = \int_0^\infty  h_1^*(\overline{F}(u_1^{-1}(t)+x))) \d t\\
%\end{align*}

{The first key observation is that, in the heavy-tailed setting, the risk measure $x_{\tau}$ is an increasing function of $\tau$ and diverges to $+\infty$. Along with mild conditions for existence and uniqueness of $x_{\tau}$ as a solution of~\eqref{eq-eqREDU-2}, this is the essential message of the following result, in which we say that a random variable $X$ is nondegenerate if it is not constant.
\begin{proposition}
\label{prop-xtau}
Let $x_{\star}=\inf\{ x\in \mathbb{R}, F(x)>0 \}$ and $x^{\star}=\inf\{ x\in \mathbb{R}, F(x)\geq 1 \}$ denote the left and right endpoints of $X$, assumed to be a nondegenerate random variable, so that $x_{\star}<x^{\star}$.
\begin{enumerate}[label=(\roman*)] 
    \item If the quantities $\mathrm{H}_{u_{1},\,h_{1}}((X-x)_{+})$ and $\mathrm{H}_{u_{2},\,h_{2}}((X-x)_{-})$ define continuous finite functions of $x\in (x_{\star},x^{\star})$ then, for any $\tau\in (0,1)$, Equation~\eqref{eq-eqREDU-2} has a unique and finite solution $x_{\tau}$.
    \item If the quantities $\mathrm{H}_{u_{1},\,h_{1}}((X-x)_{+})$ and $\mathrm{H}_{u_{2},\,h_{2}}((X-x)_{-})$ are finite when $x\in (x_{\star},x^{\star})$ and Equation~\eqref{eq-eqREDU-2} has a unique and finite solution $x_{\tau}$, then $x_{\tau}<x^{\star}$ when $\tau<1$, $\tau\in (0,1) \mapsto x_{\tau} \in \mathbb{R}$ is nondecreasing, and $\lim_{\tau\uparrow 1} x_{\tau} = x^{\star}$.
    \item Suppose that: 
    \begin{itemize}
        \item $\overline{F}\in \mathrm{RV}_{-1/\gamma}$ with $\gamma>0$, 
        \item $u_{1}$ is continuous on $[0,\infty)$ and $u_{1}\in \mathrm{RV}_{\alpha_{1}}$ with $\alpha_{1}>0$, 
        \item $1-h_{1}( 1-1/\cdot )\in \mathrm{RV}_{-\beta_{1}}$ with $\beta_{1}>0$, 
        \item $u_{2}$ is continuous on $[0,\infty)$ and $u_{2}\in \mathrm{RV}_{\alpha_{2}}$ with $\alpha_{2}>0$. 
    \end{itemize}
    Assume further that $\beta_{1}/\gamma>\alpha_{1}$ and $\int_{-\infty}^{\infty}|z|^{\alpha_{2}+\delta}\mathrm{\,d}h_{2}(F(z))<\infty$ for some $\delta>0$. Then $x_{\tau}$ exists and is unique for any $\tau\in (0,1)$, $\tau\in (0,1) \mapsto x_{\tau} \in \mathbb{R}$ is nondecreasing, and $\lim_{\tau\uparrow 1} x_{\tau} = +\infty$.
\end{enumerate}
\end{proposition}
\begin{remark}
	\label{rmk-existence} Under the regular variation conditions in (iii) above, $\beta_{1}/\gamma\geq \alpha_{1}$ is a necessary condition for the existence of $\mathrm{H}_{u_{1},\,h_{1}}((X-x)_{+})$. Indeed 
	\[
	\mathrm{H}_{u_{1},\,h_{1}}((X-x)_{+})=-\int_{z=0}^{\infty}u_{1}(z)\mathrm{\,d}%
(1-h_{1}(1-1/(1/\overline{F}(x+z)))). 
	\]
	Since $u_{1}\in \mathrm{RV}_{\alpha_{1}}$ and $ 1-h_{1}(1-1/(1/\overline{F}(x+\cdot)))\in \mathrm{RV}_{-\beta_{1}/\gamma}$, it follows that  $\mathrm{H}_{u_{1},\,h_{1}}	((X-x)_{+})$ can only be finite if $\alpha_1-\beta_1/\gamma\leq 0$, that is, $\beta_{1}/\gamma\geq \alpha_{1}$.
	\vskip1ex
	\noindent
    Besides, the assumption that $\int_{-\infty}^{\infty}|z|^{\alpha_{2}+\delta}\mathrm{\,d}h_{2}(F(z))<\infty$ is required to control the left tail behavior of the function $F$. An inspection of the proof of Lemma~\ref{lm-sides}(ii) reveals that this assumption can be weakened to $\int_{-\infty}^{\infty}|z|^{\alpha_2}\mathrm{\,d}h_{2}(F(z))<\infty$ if $u_2(y)$ is asymptotically equivalent to a multiple of $y^{\alpha_2}$ as $y\to\infty$. For $u_2(y)=y$ and $h_2(x)=x$, we find the condition  {$\mathbb{E}(|X|)<\infty$}, which is exactly the condition necessary and sufficient for the existence of expectiles.
\end{remark}
In the remainder of this paper we implicitly assume that the problem of which $x_{\tau}$ is solution is indeed well-defined and has a unique finite solution; by Proposition~\ref{prop-xtau}, if appropriate regular variation conditions on the functions involved are met, this will be guaranteed by simply assuming that $u_1$ and $u_2$ are continuous on $[0,\infty)$. We proceed by deriving, under this regularity assumption,} the asymptotic expansions of each side of \eqref{eq-eqREDU-2}, as $x\rightarrow\infty$. 
\begin{lemma}
\label{lm-sides}
Assume that $\overline{F}\in \mathrm{RV}_{-1/\gamma}$ with $\gamma>0$ and denote by $\mathrm{B}(\cdot,\cdot)$ the Beta function, that is, $\mathrm{B}(a,b)=\int_0^1 z^{a-1} (1-z)^{b-1}  \, \mathrm{d}z$ for any $a,b>0$. 
\begin{enumerate}[label=(\roman*)] 
    \item Assume $u_{1}\in \mathrm{RV}_{\alpha_{1}}$ for $\alpha_{1}>0$, $1-h_{1}( 1-1/\cdot )\in \mathrm{RV}_{-\beta_{1}}$ with $\beta_{1}>0$, and $\beta_{1}/\gamma>\alpha_{1}$. Then 
    \[
    \lim_{x\rightarrow \infty}\frac{\mathrm{H}_{u_{1},\,h_{1}}((X-x)_{+})}{\left( 1-h_{1}(F(x))\right)  u_{1}(x)}= \frac{\beta_{1}}{\gamma} \mathrm{B}(\beta_{1}/\gamma-\alpha_{1},\alpha_{1}+1) =: \Delta_{0}.
    \]
    \item Assume $u_{2}\in \mathrm{RV}_{\alpha_{2}}$ with $\alpha_{2}>0$, and
$\int_{-\infty}^{\infty}|z|^{\alpha_{2}+\delta}\mathrm{\,d}h_{2}(F(z))<\infty$
for some $\delta>0$. Then
	\[
	\lim_{x\to\infty} \frac{\mathrm{H}_{u_{2},\,h_{2}}((X-x)_{-})}{u_{2}(x)} = 1.
	\]
\end{enumerate}
\end{lemma}

% \begin{lemma}
% 	\label{left1} \label{Left}
% 	Assume that $u_{1}\in \mathrm{RV}_{\alpha_{1}}$ for
% 	$\alpha_{1}>0$, $1-h_{1}\left( 1-\frac{1}{\cdot }\right)\in \mathrm{RV}_{-\beta_{1}}$ with $\beta_{1}>0$,
% 	and $\overline{F}\in \mathrm{RV}_{-1/\gamma}$ with $\gamma>0$. Suppose
% 	$\beta_{1}/\gamma>\alpha_{1}$. Then, if $\mathrm{B}(\cdot,\cdot)$ denotes the beta function,
% \[
% \lim_{x\rightarrow \infty}\frac{\mathrm{H}_{u_{1},\,h_{1}}((X-x)_{+})}{\left(
% 1-h_{1}(F(x))\right)  u_{1}(x)}= \frac{\beta_{1}}{\gamma} \mathrm{B}(\beta_{1}/\gamma-\alpha
% _{1},\alpha_{1}+1) =: \Delta_{0}.
% \]
% \end{lemma}

% Next, we derive the asymptotic expansion of right-hand side of \eqref{eq-eqREDU-2} as $x\rightarrow\infty$.
% \begin{lemma}
% 	\label{Right1}Assume that $u_{2}\in \mathrm{RV}_{\alpha_{2}}$ with $\alpha
% 	_{2}>0$, $1-h_{2}\left( 1-\frac{1}{\cdot }\right)\in \mathrm{RV}_{-\beta_{2}}$ with $\beta_{2}>0$ and
% 	$\overline{F}\in \mathrm{RV}_{-1/\gamma}$ with $\gamma>0$. Suppose $\beta
% 	_{2}/\gamma>\alpha_{2}$. Then
% 	\[
% 	\lim_{x\to\infty} \frac{\mathrm{H}_{u_{2},\,h_{2}}((X-x)_{-})}{u_{2}(x)} = 1.
% 	\]
% \end{lemma}

We are now in position to state the first-order asymptotic expansion of the generalized shortfall risk measure.

\begin{theorem}
	\label{theo-first} Assume that $u_{1}\in \mathrm{RV}_{\alpha_{1}}$ with $\alpha_{1}>0$,
$1-h_{1}( 1-1/\cdot ) \in \mathrm{RV}_{-\beta_{1}}$ with
$\beta_{1}>0$, $u_{2}\in \mathrm{RV}_{\alpha_{2}}$ with $\alpha_{2}>0$,
$\overline{F}\in \mathrm{RV}_{-1/\gamma}$ with $\gamma>0$. Define a function $\varphi$ as
\begin{equation*}
\varphi(x)= \frac{u_{2}(x)}{u_{1}(x)(1-h_{1}(F(x)))}. %\label{phi}%
\end{equation*}
Then $\varphi\in \mathrm{RV}_s$, with $s=\alpha_2-\alpha_1+\beta
_{1}/\gamma$. Assume henceforth that $s>0$. Then $\varphi(x)$ diverges to $+\infty$ as $x\to+\infty$ and its generalized inverse function 
\begin{equation}
\varphi^{\leftarrow}:q\mapsto \inf\{x:\varphi(x)\geq q\},\quad q\in
(0,1)\label{phi-inv}%
\end{equation}
is well-defined. Further assume that
$\beta_{1}/\gamma>\alpha_{1}$ and $\int_{-\infty}^{\infty}|z|^{\alpha
_{2}+\delta}\mathrm{\,d}h_{2}(F(z))<\infty$ for some $\delta>0$. Then the
first-order expansion of the shortfall risk measure  {is
\[
{x_{\tau}}= \left[ \frac{\beta_{1}}{\gamma} \mathrm{B}(\beta_{1}/\gamma-\alpha_{1},\alpha_{1}+1) \right]^{1/s} \varphi^{\leftarrow}((1-\tau)^{-1}) (1+o(1)) = \Delta_{1}\varphi^{\leftarrow}((1-\tau)^{-1}) (1+o(1)).
\]
[In other words,  $\Delta_{1}=\Delta_{0}^{1/s}$, with $\Delta_{0}$ defined in Lemma~\ref{lm-sides}(i).]}
\end{theorem}

\begin{remark}
	\label{rmk-theo-first} {When $u_1=u_2$ and $h_1(x)=x$, the function $\varphi^{\leftarrow}$ is nothing but the tail quantile function $U$. In this case Theorem~\ref{theo-first} directly connects $x_{\tau}$ to extreme quantiles of $X$. This is reminiscent of the kind of asymptotic proportionality relationships obtained for $L^p-$quantiles, see {\it e.g.}~\cite{daouia2019Lpestimation}.}
\end{remark}

%\begin{lemma} 
%\label{lem:phiinvVaR}
%Suppose that  $1-h_{1}%
%\in\mathrm{RV}_{\beta_{1}}^{1}$ with $\beta_{1}>0$, and  $\overline{F}\in\mathrm{RV}_{-\gamma}$ with $\gamma>0$. 
%\begin{itemize}
%	
%\item If $u_1(x)/u_2(x)\to c\in\R_+$,  then we have
%$$
%   \frac{\varphi^{-1}(1-\tau)}{{\rm VaR}_{h^{-1}_1(\tau)}(X)} =   c^{1/(\beta\gamma)}~~{\rm as}~~\tau\to1.
%$$
%%\item If $u_1(x)/u_2(x)\sim c x^{s}$ as $x\to\infty$,  then we have
%%$$
%%   \frac{\varphi^{-1}(1-\tau)}{{\rm VaR}_{h^{-1}_1(\tau)}(X)} =   c^{1/(\beta\gamma)}~~{\rm as}~~\tau\to1.
%%$$
%\end{itemize}
%\end{lemma}
%
%\begin{lemma}
%Suppose that $\phi\in {\rm 2RV}_{-\alpha,\rho_1}$ with auxiliary function $A$ and $\Fbar \in {\rm 2RV}_{-\beta,\rho_2}$ with auxiliary function $B$. Then we have
%$$
%   \frac{\phi(x)}{\Fbar(x)} = c x^{\beta-\alpha} \left( 1+ \left[\frac{A(x)}{\rho_1} -\frac{B(x)}{\rho_2}\right](1+o(1))\right)
%$$
%\end{lemma}

\section{Second-order expansions of  generalized shortfall risk measures}\label{second-order}
In this section, we study the second-order asymptotics of generalized shortfall risk measures. Again, all the proofs are relegated to Section \ref{proof}.
We first prepare a few  {assumptions and} lemmas. The first lemma is regarding a set of uniform
inequalities for 2RV functions. It plays a key role in the later proofs. Moreover, it
is also an interesting result on its own as it is a complement to the usual inequalities 
on 2RV by  {providing uniform inequalities in a neighborhood of $0$.}
\begin{lemma}
	\label{bd-2RV}Assume that $g\in\mathrm{2RV}_{\gamma,\rho}$,  {with $\gamma>0$, $\rho<0$ and auxiliary function $B$, is such that $t^{-\gamma}g(t)$ is bounded on intervals of the form $(0,t_0]$, with $t_0>0$.} There exists $\widetilde{B}\sim B$ such that for any $\epsilon>0$ and
		$\delta>0$, there exists $c>0$ and $t_{0}$ such that for all $t\geq t_{0}$ and
		$0<v<\delta$,
		\[
		\left\vert \frac{\frac{g({vt})}{g(t)}-v^{\gamma}}{\widetilde{B}(t)}\right\vert
		\leq-\frac{v^{\gamma}}{\rho}(1+cv^{\rho-\epsilon}).
		\]
     {[A fixed choice of $c>0$ is possible for $\delta \in (0,1)$.]}
%		
%		\begin{comment}
%	
%		\item[(ii)] For an ultimately negative or positive function $A$ satisfying that
%		$A\in\mathrm{RV}_{\beta}$,  {with $\beta<0$, $\beta\neq \rho$,} we have
%		\[
%		\left\vert \frac{\frac{g(t(v+A(t)))}{g(t)}-v^{\gamma}}{D(t)}\right\vert
%		\leq\gamma \max\{|v-\epsilon|^{\gamma-1},1\} -\frac{(v+1)^{\gamma}}{\rho}(1+c(v-1)^{\rho-\epsilon}),
%		\]
%	where $D(t)=A(t)$ if $A(t)/\widetilde{B}(t)\rightarrow \infty$  {(that is, $\rho<\beta$)}, or
%$D(t)=\widetilde{B}(t)$ if $A(t)/\widetilde{B}(t)\rightarrow0$  {(that is, $\beta<\rho$)}.
%	\end{enumerate}
%		
%		\end{comment}
	%
\end{lemma}
 {We next present the second-order conditions about $U$, $u_1$, $u_2$ and $h_1$ that we require to obtain the second-order asymptotics of $x_{\tau}$.}
\begin{assumption}
	\label{A1} $U\in2\mathrm{RV}_{\gamma,\rho}$ for $\gamma>0$ and $\rho<0$
	with auxiliary function $A(t)$.
\end{assumption}
\begin{assumption}
	\label{A2} For $i=1,2$, $u_{i}\in2\mathrm{RV}_{\alpha_{i},\eta_{i}}$
	for $\alpha_{i}>0$ and $\eta_{i}<0$ with auxiliary function $B_{i}(t)$,  {and $t^{-\alpha_i} u_i(t)$ is bounded on intervals of the form $(0,t_0]$, for $t_0>0$.}
\end{assumption}
\begin{assumption}
	\label{A3} $1-h_{1}(1-1/\cdot)\in2\mathrm{RV}_{-\beta_{1}%
		,\varsigma}$ for $\beta_{1}>0$ and $\varsigma<0$ with auxiliary function
	$C(t)$, and  {$1-h_{2}(1-1/\cdot)\in \mathrm{RV}_{-\beta_{2}}$ for $\beta_{2}>0$.}
\end{assumption}
Under Assumptions \ref{A1}, \ref{A2}, and \ref{A3}, by Lemma~\ref{lm-2RV-rep-inv}, and Propositions 2.6 and 2.9 in \cite{lv2012properties}, we
immediately obtain the following useful results.

\begin{lemma}\label{notation}
	Under Assumptions \ref{A1}, \ref{A2}, and \ref{A3},
	
	\begin{enumerate}[label=(\roman*)] 
		\item $U$ has the representation, as $x\rightarrow\infty$,%
		\[
		U(x)=cx^{\gamma}\left[  1+\frac{1}{\rho}A(x)+o(A(x))\right],  {\mbox{ where } c>0.}
		\]
		Consequently,  {$\overline{F}(\cdot) \in2\mathrm{RV}_{-1/\gamma,\rho/\gamma}$ with auxiliary
		function $A_{F}(t)=\gamma^{-2}A(1/\overline{F}(t))$, and $\overline{F}(\cdot)$ has the 
		representation,} as $x\rightarrow\infty$,
		\[
		\overline{F}(x)=c^{1/\gamma}x^{-1/\gamma}\left[  1+\frac{1}{\gamma\rho
		}A(1/\overline{F}(x))+o(A(1/\overline{F}(x)))\right]  .
		\]
%		
%		
%		\item $F^{\leftarrow}(\cdot)$ has the representation, as $\tau\rightarrow1$,%
%		\[
%		F^{\leftarrow}(\tau)=U\left(  \frac{1}{1-\tau}\right)  =c(1-\tau)^{-\gamma
%		}\left[  1+\frac{1}{\rho}A\left(  \frac{1}{1-\tau}\right)  +o\left(  A\left(
%		\frac{1}{1-\tau}\right)  \right)  \right]  .
%		\]

		\item For $i=1,2$, $u_{i}$ has the representation, as $x\rightarrow\infty$,
		\[
		u_{i}(x)=a_{i}x^{\alpha_{i}}\left[  1+\frac{1}{\eta_{i}}B_{i}(x)+o(B_{i}%
		(x))\right],  {\mbox{ where } a_i>0.}
		\]

		\item $1-h_{1}(1-1/\cdot)$ has the representation, as $x\rightarrow\infty$,%
		\[
		1-h_{1}\left(  1-\frac{1}{x}\right)  =bx^{-\beta_{1}}\left[  1+\frac
		{1}{\varsigma}C(x)+o(C(x))\right],  {\mbox{ where } b>0.}
		\]

		\item  {$1-h_{1}(  F(\cdot))$ has the representation, as $x\rightarrow\infty$,
		\[
		1-h_{1}(  F(x)) = bc^{\beta_{1}/\gamma}x^{-\beta_{1}/\gamma}\left[  1+\frac{\beta_{1}%
        }{\gamma \rho}A(1/\overline{F}(x))(  1+o(1)) +\frac{1}{\varsigma}C(1/\overline{F}(x))(1+o(1))\right].
		\]
		In particular, if $C(x)/A(x)\to \kappa\in [-\infty,+\infty]$ as $x\to \infty$ with $\kappa\neq -\beta_1/\gamma$, then $A_{h}(\cdot)=\gamma^{-1}((\beta_{1}/\gamma)A(1/\overline{F}(\cdot))+C(1/\overline{F}(\cdot)))$ is nonzero and has constant sign in a neighborhood of infinity, $|A_{h}(\cdot)|$ is regularly varying with index $\rho_{h}=\max\left\{  \rho,\varsigma\right\}  /\gamma$, and $1-h_{1}(  F(\cdot))  \in2\mathrm{RV}_{-\beta_{1}/\gamma
			,\rho_{h}}$ for $\rho_{h}=\max\left\{  \rho,\varsigma\right\}  /\gamma$ with
	    auxiliary function $A_h$.}
     \item  {$1-h_{2}(F(\cdot))\in \mathrm{RV}_{-\beta_{2}/\gamma}$.}
	%	
	%	\item  {In the setting of (iv), the left-continuous inverse $U_{h_{1}}$ of $1/(1-h_1(F(\cdot))$ satisfies} $U_{h_{1}}(\cdot)\in2\mathrm{RV}_{\gamma/\beta_{1},\rho_{U}}$ for
	%	$\rho_{U}=\max\left\{  \rho,\varsigma\right\}  /\beta_{1}$ with  {auxiliary}
	%	function $A_{U}(t)=(\beta_{1}/\gamma)^{-2}A_{h}(U_{h_{1}}(t))$.
	\end{enumerate}
\end{lemma}

\begin{remark}
	\label{rmk-notation}  {Condition $C(x)/A(x)\to \kappa\in [-\infty,+\infty]$ as $x\to\infty$ with $\kappa\neq -\beta_1/\gamma$ in Lemma~\ref{notation}(iv) is very mild. It is in particular satisfied as soon as $\rho\neq \varsigma$, corresponding to the case when either $A$ or $C$ dominates in $A_h$. When $\rho=\varsigma$, in typical second-order regularly varying models $A$ and $C$ will be proportional to the same negative power function $t\mapsto t^{\rho}$, and the condition simply says that the proportionality constants should not cancel in the calculation of $A_h$. If this condition is not satisfied, then $1-h_{1}(  F(\cdot))$ would typically still be second-order regularly varying, but the second-order parameter and auxiliary function would depend on the {\it third-order} regular variation properties of $\overline{F}$ and $1-h_{1}(1-1/\cdot)$.}
\end{remark}

 {The next lemma analyzes the second-order regular variation properties of the left-continuous inverse function $\varphi^{\leftarrow}$ defined in (\ref{phi-inv}) and its connection with extreme quantiles of the distribution function $F$. It will be used in the proof of the main result of this section.}
\begin{lemma}
\label{lm-39-1} Under Assumptions \ref{A1}, \ref{A2}, and \ref{A3},  {and if there is a regularly varying function $D$ such that $A(1/\overline{F}(x))/D(x)\to a\in \mathbb{R}$, $B_i(x)/D(x)\to b_i\in \mathbb{R}$ and $C(x)/D(x)\to \kappa\in \mathbb{R}$ as $x\to\infty$, with $b_2/\eta_2-b_1/\eta_1-(a\beta_1/\gamma+\kappa)/(\gamma\rho_h)\neq 0$, we have, as $\tau \rightarrow 1$,
\[
\varphi^{\leftarrow}(\left(  1-\tau \right)  ^{-1})=c^{\ast}(1-\tau
)^{-1/s}\left(  1-\frac{1}{s}A^{\ast}(\varphi^{\leftarrow}(\left(
1-\tau \right)  ^{-1}))(1+o(1))\right)
\]
%$\varphi^{-1}\in \mathrm{2RV}%
%_{-s,\eta^{\ast}}$ for $s=\alpha_{1}-\alpha_{2}-\beta_{1}/\gamma$ and
%$\eta^{\ast}=\max \{ \eta_{1},\rho,\eta_{2}\}$ with auxiliary function $A^{\ast
%}(x)=-A(x)-B(x)+C(x)$, and $\varphi^{\leftarrow}\in \mathrm{2RV}_{1/s,\eta
%^{\ast}/s}^{0}$ with auxiliary function $-s^{-2}A^{\ast}\circ \varphi
%^{\leftarrow}$.%
where $s=\alpha_{2}-\alpha_{1}+\beta_{1}/\gamma$ as in Theorem~\ref{theo-first}, and $c^{\ast}=\left(
{\frac{a_{2}}{bc^{\beta_{1}/\gamma}a_{1}}}\right)  ^{-1/s}$ and $A^{\ast
}(t)=\frac{1}{\eta_{2}}B_{2}(t)-\frac{1}{\eta_{1}}B_{1}(t)-\frac{1}{\rho_{h}%
}A_{h}(t)$ is regularly varying with index $\eta^{\ast
}=\max \{ \eta_{1},\rho_{h},\eta_{2}\}$, with the notation of Lemma~\ref{notation}. In particular, $\varphi^{\leftarrow}\in \mathrm{2RV}_{1/s,\eta^{\ast}/s}$ and
\[
\frac{\varphi^{\leftarrow}(\left(  1-\tau \right)  ^{-1})}{\left(
F^{\leftarrow}(\tau)\right)  ^{1/(\gamma s)}}=c_{0}\left(  1- \frac
{c_{0}^{\eta^*}}{s}A^{\ast}(( F^{\leftarrow}(\tau) )^{1/(\gamma
s)})(1+o(1))-\frac{1}{\gamma s\rho}A( (1-\tau)^{-1} ) (1+o(1))
\right)
\]
where $c_{0}=c^{\ast}c^{-1/(\gamma s)}$. In the specific setting when $u_1=u_2$, the condition linking $A$, the $B_i$, $C$ and $a$, $b_1$, $b_2$ and $\kappa$ can be replaced by supposing that $C(x)/A(x)\to \kappa\in [-\infty,+\infty]$ as $x\to \infty$ with $\kappa\neq -\beta_1/\gamma$, in which case $A^*=-\frac{1}{\rho_{h}%
}A_{h}$.}
\end{lemma}

 {To} derive the second-order asymptotic expansions for the generalized shortfall risk
measure, we proceed by analyzing the two sides of (\ref{eq-eqREDU-2}) separately.

\begin{lemma}
\label{left2}
Under Assumptions \ref{A1}, \ref{A2}, and \ref{A3},  {further assume that $C(x)/A(x)\to \kappa\in [-\infty,+\infty]$ as $x\to \infty$ with $\kappa\neq -\beta_1/\gamma$, as well as $\beta_{1}/\gamma>\alpha_{1}$ and $\alpha_{1}+\eta_{1}>0$}. Then
as $x\rightarrow \infty$,
\[
\frac{\mathrm{H}_{u_{1},\,h_{1}}((X-x)_{+})}{\left(  1-h_{1}(F(x))\right)
u_{1}(x)}=\Delta_{0}+\Gamma_{1}B_{1}(x)(1+o(1))+\Gamma_{2}A_{h}(x)(1+o(1))
\]
with
\[
 {\Gamma_{1}=\frac{\beta_1}{\gamma} \times \frac{1}{\eta_1} ( \mathrm{B}(\beta_1/\gamma-\alpha_1-\eta_1,\alpha_1+\eta_1+1) - \mathrm{B}(\beta_1/\gamma-\alpha_1,\alpha_1+1) )} %=\frac{1}{\eta_1} \left( (\alpha_1+\eta_1) \mathrm{B}(\beta_1/\gamma-\alpha_1-\eta_1,\alpha_1+\eta_1) - \alpha_1 \mathrm{B}(\beta_1/\gamma-\alpha_1,\alpha_1) \right) %\frac{1}{\eta_{1}}\int_{0}^{1}\left(  v^{-\gamma/\beta_{1}}-1\right)  ^{\alpha_{1}+\eta_{1}}-\left(  v^{-\gamma/\beta_{1}}-1\right)^{\alpha_{1}}\mathrm{\,d}v,
\]
and
\[
 {\Gamma_{2}= \frac{1}{\rho_h} \left( \left( \frac{\beta_1}{\gamma} - \rho_h \right) \mathrm{B}(\beta_1/\gamma-\alpha_1-\rho_h,\alpha_1+1) - \frac{\beta_1}{\gamma} \mathrm{B}(\beta_1/\gamma-\alpha_1,\alpha_1+1) \right)}. %\frac{\alpha_{1}}{\rho_{h}}\int_{0}^{1}\left(  v^{-\rho_{h}}-1\right)  (1-v)^{\alpha_1-1}v^{\beta_{1}/\gamma-\alpha_1-1}\mathrm{\,d}v.
\]

\end{lemma}

Now we turn to the right-hand side of (\ref{eq-eqREDU-2}).

\begin{lemma}
\label{Right2}
Assume that $u_{2}$ is differentiable and $u_{2}^{\prime}\in \mathrm{RV}%
_{\alpha_{2}-1}$  {is bounded on finite intervals of the form $(0,t_0]$ ($t_0>0$),} with either $\alpha_{2}>1$ or $\alpha_{2}=1$ and
$u_{2}^{\prime}$ nondecreasing, and $\overline{F}\in \mathrm{RV}_{-1/\gamma}$
with $\gamma>0$. Suppose $\int_{-\infty}^{\infty}|z|^{\alpha_{2}+\delta
}\mathrm{\,d}h_{2}(F(z))<\infty$ for some $\delta>0$. Then as $x\rightarrow
\infty$,
\[
 {\frac{\mathrm{H}_{u_{2},\,h_{2}}((X-x)_{-})}{u_{2}(x)}=1-(1-h_{2}%
(F(x)))-x^{-1}(\alpha_{2}\mathbb{E}[Z]+o(1)),}
\]
where $\mathbb{E}[Z]=\int_{-\infty}^{\infty}z\mathrm{\,d}h_{2}(F(z))$.  {[The random variable $Z$ has distribution function $h_2(F(\cdot))$.]}

\end{lemma}

%\begin{align*}
%{\rm H}_{u_1, \, h_1} ((X-x)_+)& = \int_x^\infty u_1(y-x) \d h_1(F(y))\\
%& = \E[u_1(Y-x)_+] = \E[u_1(U(1/t)-x)_+] \\
%&  = \int_0^\infty \p(u_1((Y-x)_+) > t) \d t\\
%& = \int_0^\infty \p((Y-x)_+  > u_1^{-1}(t)) \d t\\
%& = \int_0^\infty \p( Y-x   > u_1^{-1}(t)) \d t\\
%& = \int_0^\infty \p( Y   > u_1^{-1}(t)+x) \d t\\
%& = \int_0^\infty 1-h_1(F(u_1^{-1}(t)+x))) \d t\\
%& = \int_0^\infty  h_1^*(\overline{F}(u_1^{-1}(t)+x))) \d t\\
%\end{align*}
%\begin{align*}
%{\rm H}_{u_1, \, h_1} ((X-x)_+)& = \int_x^\infty u_1(y-x) \d h_1(F(y))\\
%& = \int_{1/\overline{F}(x)}^\infty u_1(U(t)-x) \d h_1(1-1/t).
%\end{align*}
%Then as $x\to\infty$, we have
%\begin{align*}
%\frac{{\rm H}_{u_1, \, h_1} ((X-x)_+)}{u_1(x)}
%& = \int_x^\infty \frac{u_1(U(1/t)-x)}{u_1(x)} \d h_1(t)\\
%& = \int_x^\infty \frac{u_1(y-x)}{u_1(U(1/t))} \d h_1(F(y))\\
%& = \int_1^\infty \frac{u_1(yx-x)}{u_1(x)} \d h_1(F(yx))\\
%\end{align*}
%\newpage
%\begin{equation}
%\label{eq-eqEU-3}
%\rho_\alpha(X)=\inf\{x\in\R:\alpha \, \E[\ell_1((X-x)_+)] \le (1-\alpha) \E[\ell_2((X-x)_-)]\}
%\end{equation}
%\begin{equation}
%\label{eq-eqEU-3}
%\rho_\alpha(X)=\inf\{x\in\R: \E[\ell_1((X-x)_+)] \le \frac{1-\alpha}{\alpha} \E[\ell_2((X-x)_-)]\}
%\end{equation}

\begin{remark}
In Lemma \ref{Right2}, the tail index $\alpha_{2}$ is restricted to be greater
than $1$. This is because if $\alpha_{2}<1$, then additional conditions are
needed to ensure $h_{2}(F(\cdot))$ is regularly varying at $0$. For
simplicity, we omit this case.  {The assumption that $u_2'$ is bounded on finite intervals of the form $(0,t_0]$ essentially amounts to assuming that $t\mapsto t^{-\alpha_2} u_2(t)=L_2(t)$ is smooth in a neighborhood of 0 and $t\mapsto L_2(t)/t$ is bounded. It therefore intuitively represents a strengthened version of part of Assumption~\ref{A2}.} 
\end{remark}

Next, we present the second-order expansion of ${x_{\tau}}$ in terms of
$\varphi^{\leftarrow}(\left(  1-\tau \right)  ^{-1})$,  {obtained essentially by combining Lemmas~\ref{left2} and~\ref{Right2}.}

\begin{theorem}
\label{main} Under Assumptions \ref{A1}, \ref{A2}, and \ref{A3}, further
 {assume that there is a regularly varying function $D$ such that $A(1/\overline{F}(x))/D(x)\to a\in \mathbb{R}$, $B_i(x)/D(x)\to b_i\in \mathbb{R}$ and $C(x)/D(x)\to \kappa\in \mathbb{R}$ as $x\to\infty$, with $b_2/\eta_2-b_1/\eta_1-(a\beta_1/\gamma+\kappa)/(\gamma\rho_h)\neq 0$. Suppose also that $u_{2}$ is differentiable and $u_{2}^{\prime}\in \mathrm{RV}%
_{\alpha_{2}-1}$ is bounded on finite intervals of the form $(0,t_0]$}, with either $\alpha_{2}>1$ or $\alpha_{2}=1$ and
$u_{2}^{\prime}$ nondecreasing. Suppose $\beta_{1}/\gamma>\alpha_{1}$,
 {$\alpha_{1}+\eta_{1}>0$} and $\int_{-\infty}^{\infty}|z|^{\alpha_{2}+\delta
}\mathrm{\,d}h_{2}(F(z))<\infty$ for some $\delta>0$. We have,  {as
$\tau \rightarrow1$,
\begin{align*}
& \frac{x_{\tau}}{\Delta_{1}\varphi^{\leftarrow}((1-\tau)^{-1})}-1\\
&  =\frac{1}{s}\left(  \frac{\Gamma_{1}}{\Delta_{0}^{1-\eta_{1}/s}}%
B_{1}(\varphi^{\leftarrow}((1-\tau)^{-1}))(1+o(1))+\frac{\Gamma_{2}}%
{\Delta_{0}^{1-\rho_{h}/s}}A_{h}(\varphi^{\leftarrow}((1-\tau)^{-1}%
))(1+o(1))\right. \\
&  +\Delta_{0}^{-\beta_{2}/(\gamma s)}(  1-h_{2}(F(\varphi^{\leftarrow
}((1-\tau)^{-1}))) )  (1+o(1))+\frac{\alpha_{2}\Delta
_{0}^{-1/s}}{\varphi^{\leftarrow}((1-\tau)^{-1})}(\mathbb{E}[Z]+o(1))\\
&  \left.  -(\Delta_{0}^{\eta^{\ast}/s}-1)A^{\ast}(\varphi^{\leftarrow
}((1-\tau)^{-1}))(1+o(1))-(1-\tau)(1+o(1))\right)
\end{align*}
with the notation of the above lemmas.}
\end{theorem}

Combining Lemma \ref{lm-39-1} and Theorem \ref{main}, we finally obtain the  {desired second-order expansion} of ${x_{\tau}}$ in terms of $F^{\leftarrow}(\tau)$.

\begin{theorem}
\label{complete} Under the conditions of Theorem \ref{main}, we have, as $\tau \rightarrow1$  {and with $c_0$ as in Lemma~\ref{lm-39-1}, 
\begin{align*}
&  \frac{x_{\tau}}{c_{0}\Delta_{1}(F^{\leftarrow}(\tau))^{1/(\gamma s)}}-1\\
&  =\Delta_{2}B_{1}((F^{\leftarrow}(\tau))^{1/(\gamma s)})(1+o(1))+\Delta
_{3}A_{h}((F^{\leftarrow}(\tau))^{1/(\gamma s)})(1+o(1))\\
&  +\Delta_{4}( 1-h_{2}(F((F^{\leftarrow}(\tau))^{1/(\gamma
s)})) ) (1+o(1))+(F^{\leftarrow}(\tau))^{-1/(\gamma
s)}(\Delta_{5}+o(1))\\
&  -\Delta_{6}A^{\ast}((F^{\leftarrow}(\tau))^{1/(\gamma s)})(1+o(1))-\frac
{1}{\gamma s\rho}A((1-\tau)^{-1})(1+o(1))-\frac{1}{s}(1-\tau)(1+o(1)),
\end{align*}
with $\Delta_{2}=s^{-1}\Gamma_{1}c_{0}^{\eta_{1}}\Delta_{0}^{\eta_{1}/s-1}$,
$\Delta_{3}=s^{-1}\Gamma_{2}c_{0}^{\rho_{h}}\Delta_{0}^{\rho_{h}/s-1}$,
$\Delta_{4}=s^{-1}c_{0}^{-\beta_{2}/\gamma}\Delta_{0}^{-\beta_{2}/(\gamma s)}$,
$\Delta_{5}=s^{-1}c_{0}^{-1}\alpha_{2}\mathbb{E}[Z]\Delta_{0}^{-1/s}$, and
$\Delta_{6}=s^{-1}c_{0}^{\eta^{\ast}}\Delta_{0}^{\eta^{\ast}/s}$.}
\end{theorem}
\begin{remark}
 {The auxiliary function $A(t)$ in (\ref{eq-d2}) of Definition \ref{2rvdef} is of course only unique up to asymptotic equivalence. Given a distribution function $F$ or tail quantile function $U$ of a random variable $X$, a reasonable choice of auxiliary function, readily computed, would be the function $A_{0}$ in Theorem 2.3.9 of \cite{de2006extreme}, which guarantees a uniform kind of second-order regular variation. That being said, the asymptotic expansions in Theorems \ref{main} and \ref{complete} hold true for any other choice of $A$ asymptotically equivalent to this function $A_{0}$, and similarly for the choices of $B_1,B_2$ and $C$.}
\end{remark}

\begin{corollary}
\label{coro:stat}  {Under the conditions of Theorem~\ref{main}, we have, as $\tau \rightarrow 1$,}
\begin{multline*}
 {x_{\tau} = \Delta_{1} \varphi^{\leftarrow}( ( 1-\tau )^{-1} ) \big(  1+O(( 1-\tau )^{1/\max(s,1)}) + O(1-h_2(F((1-\tau)^{-1/s})))} \\
     {+ O(A((1-\tau)^{-1/(\gamma s)})) + O(B_1(( 1-\tau )^{-1/s})) + O(B_2(( 1-\tau )^{-1/s})) + O(C((1-\tau)^{-1/(\gamma s)})) \big).}
\end{multline*}
\end{corollary}

\begin{remark}
\label{rmk:purepower}
 {Theorems~\ref{main} and~\ref{complete} and Corollary~\ref{coro:stat} also hold if either of the functions $U$, $u_i$ or $1-h_i(1-1/\cdot)$ is a multiple of a pure power function, with corresponding conditions on the second-order parameter(s) dropped and the corresponding auxiliary function(s) involved taken identically equal to 0. Such examples are considered in Section~\ref{example} below. In Corollary~\ref{coro:stat}, the first term $O(( 1-\tau )^{1/\max(s,1)})$ should in practice be understood as $O( 1-\tau )+O(1/\varphi^{\leftarrow}( ( 1-\tau )^{-1} ))$; when $u_1=u_2$, corresponding to the {\it a priori} reasonable setting in risk management when the (non-distorted) cost of a deviation of the predictor from below or above $X$ is the same, then $\varphi^{\leftarrow}$ is nothing but the tail quantile function of the (distorted) distribution function $h_1(F(\cdot))$. Terms proportional to the reciprocal of a tail quantile function are standard in asymptotic expansions of risk measures, see {\it e.g.}~\cite{daouia2018estimation} and~\cite{daouia2019Lpestimation} in the expectile and $L^p-$quantile setting. In this case, note that, as in Lemma~\ref{lm-39-1}, the condition linking $A$, the $B_i$, $C$ and $a$, $b_1$, $b_2$ and $\kappa$ can be replaced by supposing that $C(x)/A(x)\to \kappa\in [-\infty,+\infty]$ as $x\to \infty$ with $\kappa\neq -\beta_1/\gamma$.}
\end{remark}

\section{ {Estimation}}
\label{sec:estim}

 {Theorem~\ref{theo-first} provides an asymptotic equivalent of the non-explicit shortfall risk measure $x_{\tau}$ (at extreme levels) in terms of the generalized inverse of the function $\varphi$, which is obtained by simple operations on the functions $u_1$, $u_2$ and $h_1$ chosen by the user, and the unknown distribution function $F$. An estimator of $x_{\tau}$ at extreme levels can thus essentially be constructed by estimating the function $F$ at extreme levels and inverting the resulting estimated version of $\varphi$. Since the main statistical difficulty resides in the estimation of $F$, we illustrate this principle in the particular situation when $u_1=u_2=u$, and $h_1,h_2$ are continuous and strictly increasing functions with $1-h_1( 1-1/\cdot ) \in \mathrm{RV}_{-1}$. This results in the simpler setting when $\varphi(\cdot)=1/(1-h_1(F(\cdot)))$, making it possible to avoid technicalities due to the (different) regular variation properties of $u_1$, $u_2$, $h_1$ and $h_2$, and contains not only the case when $h_1=h_2$ is the identity function, for which $\varphi^{\leftarrow}$ is nothing but the tail quantile function of $X$, but also the interesting case when $1-h_1( 1-1/x )$ and $1-h_2( 1-1/x )$ are equivalent to a multiple of $1/x$ as $x\to\infty$. The former situation contains the example of $L^p-$quantiles and the latter encompasses the example of generalized expectiles, both of which will be considered in Section~\ref{example}. The general case is of course handled in much the same way, at the price of further burdensome calculations.}

 {When $u_1=u_2=u$ is regularly varying with index $\alpha>0$ and $h_1$ is such that $1-h_1( 1-1/\cdot ) \in \mathrm{RV}_{-1}$, Theorem~\ref{theo-first} suggests that 
\[
x_{\tau}=\left( \frac{1}{\gamma} \mathrm{B}(1/\gamma-\alpha,\alpha+1) \right)^{\gamma} \varphi^{\leftarrow}((1-\tau)^{-1}) (1+o(1)) \mbox{ as } \tau\uparrow 1. 
\]
Since $\varphi(\cdot)=1/(1-h_1(F(\cdot)))$, $\varphi^{\leftarrow}((1-\tau)^{-1})$ is nothing but the quantile of level $\tau$ of the random variable having distribution function $h_1(F(\cdot))$, that is, 
\begin{align*}
x_{\tau}&=  \left( \frac{1}{\gamma} \mathrm{B}(1/\gamma-\alpha,\alpha+1) \right)^{\gamma} F^{\leftarrow}(h_1^{-1}(\tau)) (1+o(1)) \\ 
        &= \left( \frac{\gamma}{\mathrm{B}(1/\gamma-\alpha,\alpha+1)} \right)^{-\gamma} F^{\leftarrow}(h_1^{-1}(\tau)) (1+o(1)) \mbox{ as } \tau\uparrow 1. 
\end{align*}
Since $h_1^{-1}(\tau)\uparrow 1$ as $\tau\uparrow 1$, the above identity shows that the problem of estimating $x_{\tau}$ for $\tau$ large reduces to estimating $\gamma$ and extreme quantiles of $F$.}

 {Suppose then that $X_1,\ldots,X_n$ is a sample of data from a distribution function $F$ such that $\overline{F}(\cdot)\in2\mathrm{RV}_{-1/\gamma,\rho/\gamma}$. The data $X_1,\ldots,X_n$ are allowed to be serially dependent. Let also $\tau_n\uparrow 1$ be an extreme level: typical interesting cases are those when $n(1-\tau_n)$ is bounded in $n$, such as $\tau_n=1-1/n$. A standard way to estimate the extreme quantile $q_{\tau_n}\equiv F^{\leftarrow}(\tau_n)$ is to use the estimator due to~\cite{wei1978}, defined as 
\[
\widehat{q}_{\tau_n} \equiv \widehat{q}_{\tau_n}(k_n) = \left( \frac{k_n}{n(1-\tau_n)} \right)^{\widehat{\gamma}_n} X_{n-k_n,n} 
\]
where $(k_n)$ is a sequence of integers tending to infinity, with $k_n/n\to 0$ and $n(1-\tau_n)/k_n\to 0$, $X_{1,n}\leq X_{2,n}\leq \cdots \leq X_{n,n}$ are the order statistics of the sample $(X_1,\ldots,X_n)$ arranged in increasing order, and $\widehat{\gamma}_n$ is an estimator of the parameter $\gamma$. A reasonable choice of $\widehat{\gamma}_n$ is the estimator of~\cite{hil1975}: 
\[
\widehat{\gamma}_n=\frac{1}{k_n} \sum_{i=1}^{k_n} \log X_{n-i+1,n}-\log X_{n-k_n,n}.
\]
We may then define the following estimator of $x_{\tau_n}$: 
\begin{align*}
\widehat{x}_{\tau_n} \equiv \widehat{x}_{\tau_n}(k_n) &= \left( \frac{1}{\widehat{\gamma}_n} \mathrm{B}(1/\widehat{\gamma}_n-\alpha,\alpha+1) \right)^{\widehat{\gamma}_n} \widehat{q}_{h_1^{-1}(\tau_n)}(k_n) \\
    &= \left( \frac{k_n}{n(1-h_1^{-1}(\tau_n))} \right)^{\widehat{\gamma}_n} \left\{ \left( \frac{1}{\widehat{\gamma}_n} \mathrm{B}(1/\widehat{\gamma}_n-\alpha,\alpha+1) \right)^{\widehat{\gamma}_n} X_{n-k_n,n}  \right\}.
\end{align*}
This is also a Weissman-type estimator of $x_{\tau_n}$. We have the following convergence result for $\widehat{x}_{\tau_n}$.
\begin{theorem}
\label{theo-weissman}
 {Assume that: 
\begin{itemize}
\item $U\in2\mathrm{RV}_{\gamma,\rho}$ for $\gamma>0$ and $\rho<0$ with auxiliary function $A$, 
\item $u_1=u_2=u\in 2\mathrm{RV}_{\alpha,\eta}$ for $\alpha>0$ and $\eta<0$ with auxiliary function $B$, and $t^{-\alpha} u(t)$ is bounded on intervals of the form $(0,t_0]$, for $t_0>0$, 
\item $u$ is differentiable and $u^{\prime}\in \mathrm{RV}_{\alpha-1}$ is bounded on finite intervals of the form $(0,t_0]$, with either $\alpha>1$, or $\alpha=1$ and $u^{\prime}$ nondecreasing, 
\item $1-h_1( 1-1/\cdot ) \in \mathrm{2RV}_{-1,\varsigma}$ for $\varsigma<0$ with auxiliary function $C$. 
\end{itemize}
Assume also that $C(x)/A(x)\to \kappa\in [-\infty,+\infty]$ as $x\to \infty$ with $\kappa\neq -1/\gamma$, and that $1/\gamma>\alpha$, $\alpha+\eta>0$ and $\int_{-\infty}^{\infty}|z|^{\alpha+\delta}\mathrm{\,d}h_{2}(F(z))<\infty$ for some $\delta>0$. Let $(k_n)$ be a sequence of integers and $(\tau_n)$ be a sequence converging to 1 such that $k_n\to\infty$, $k_n/n\to 0$, $n(1-\tau_n)/k_n\to 0$, $\log ( k_n/(n(1-\tau_n)) )/\sqrt{k_n} \to 0$ and $\sqrt{k_n} (k_n/n+|A(n/k_n)|+|B(q_{1-k_n/n})|+|C(n/k_n)|+1/q_{1-k_n/n}) = \operatorname{O}(1)$ as $n \to \infty$. If
\[
\sqrt{k_n} ( \widehat{\gamma}_n-\gamma ) \stackrel{d}{\longrightarrow} N \ \mbox{ and } \ \sqrt{k_n} \left( \frac{X_{n-k_n,n}}{q_{1-k_n/n}}-1 \right) \stackrel{d}{\longrightarrow} N'
\]
where $N$ and $N'$ are nondegenerate distributions, then
\[
\frac{\sqrt{k_n}}{\log ( k_n/(n(1-\tau_n)) )} \left( \frac{\widehat{x}_{\tau_n}}{x_{\tau_n}} - 1 \right) \stackrel{d}{\longrightarrow} N.
\]
}
\end{theorem}
Note that, following Remark~\ref{rmk:purepower}, if either of the functions $U$, $u$ or $1-h_1(1-1/\cdot)$ is a multiple of a pure power function, then Theorem~\ref{theo-weissman} holds with corresponding conditions on the second-order parameter(s) dropped and the corresponding auxiliary function(s) involved taken identically equal to 0. For instance, if $u(x)$ is proportional to $x^{\alpha}$, then $B$ can be taken equal to 0 and condition $\alpha+\eta>0$ disappears.}

 {An important subcase in which Theorem~\ref{theo-weissman} applies is when the $X_i$ are independent. In this setting, it is known that when $\sqrt{k_n} A(n/k_n)\to\lambda\in \mathbb{R}$, 
\begin{equation}
\label{eqn:convhighlevel}
\sqrt{k_n} \left( \widehat{\gamma}_n-\gamma, \frac{X_{n-k_n,n}}{q_{1-k_n/n}}-1 \right) \stackrel{d}{\longrightarrow} \left( \frac{\lambda}{1-\rho}, 0 \right) + \gamma (\Theta,\Psi)
\end{equation}
where $\Theta$ and $\Psi$ are independent standard normal random variables, as can be seen by combining Lemma~3.2.3 and Theorem~3.2.5 in~\cite{de2006extreme}. Then, by Theorem~\ref{theo-weissman}, 
\[
\frac{\sqrt{k_n}}{\log(k_n/(n(1-\tau_n)))} \left( \frac{\widehat{x}_{\tau_n}}{x_{\tau_n}} - 1 \right) \stackrel{d}{\longrightarrow} \mathcal{N}\left( \frac{\lambda}{1-\rho}, \gamma^2 \right).
\]
Extensions of convergence~\eqref{eqn:convhighlevel} to the case when the $X_i$ are serially dependent, such as when the $X_i$ are strongly mixing (namely, $\alpha-$mixing) or absolutely regular (namely, $\beta-$mixing), thus covering standard linear time series or conditionally heteroskedastic random processes, are examined in {\it e.g.}~\cite{hsi1991} and~\cite{dre2003}. In such models, just like $\widehat{\gamma}_n$, the estimator $\widehat{x}_{\tau_n}$ will still be asymptotically Gaussian but with an enlarged variance, due to the loss of information entailed by the presence of serial dependence in the data.}

\section{Examples  {and numerical illustrations}}\label{example}
In this section, we discuss two interesting examples of  {generalized shortfall risk measures, and we briefly examine the finite-sample performance of the estimator presented in Section~\ref{sec:estim}.}

{
\begin{example}($L^p$-quantiles)\label{Lp quantile}
Let $u_{1}(x)=u_{2}(x)=px^{p-1}$, $p\geq 1$, and $h_{1}(x)=h_{2}(x)=x$. Then
$x_{\tau}$ is reduced to the $L^{p}$-quantile in \cite{daouia2019Lpestimation}, denoted by $x_{\tau}^{Lp}$. We examine the first- and second-order expansions of $x_{\tau
}^{Lp}$ arising from our results when $\overline{F} \in 2\mathrm{RV}_{-1/\gamma,\rho/\gamma}$ with $\gamma>0$ and $\rho<0$.

Clearly $u_{i}\in \mathrm{RV}_{p-1}$ and $1-h_{i}(1-1/\cdot)\in \mathrm{RV}_{-1}$ for $i=1,2$. Conditions $1/\gamma>p-1$ and $\int_{-\infty}^{\infty} |z|^{p-1+\delta} \mathrm{d}h_2(F(z))<\infty$ for some $\delta>0$ reduce to $\gamma<1/(p-1)$ and $\mathbb{E}(|\min(X,0)|^{p-1+\delta})<\infty$ (the latter can be replaced by $\mathbb{E}(|\min(X,0)|^{p-1})<\infty$, see Remark~\ref{rmk-existence}). The function $\varphi$ is nothing but $1/\overline{F}$, so $\varphi((1-\tau)^{-1})=F^{\leftarrow}(\tau)$. By Theorem~\ref{theo-first}, the first-order asymptotic expansion of $x_{\tau}^{Lp}$ is
\[
x_{\tau}^{Lp}=\Delta_{1} F^{\leftarrow}(\tau)(1+o(1)) = \left( \frac{1}{\gamma} \mathrm{B}(1/\gamma-p+1,p) \right)^{\gamma} F^{\leftarrow}(\tau)(1+o(1)) \ \mbox{ as } \tau\uparrow 1. 
\]
This recovers Corollary 1 of \cite{daouia2019Lpestimation}.

We now analyze the second-order expansion provided by Theorem~\ref{main} when $p\geq 2$, in which case $X$ has a finite moment of order 1 under our assumptions. Obviously the $u_{i}$ and $1-h_{i}( 1-1/\cdot )$ are multiples of pure power functions, and $1-h_{1}( F(\cdot) ) = 1-F(\cdot) \in 2\mathrm{RV}_{-1/\gamma,\rho/\gamma} $ with auxiliary function $\gamma^{-2}A(1/\overline{F}(\cdot))$. In other words, with the notation of Theorem~\ref{main}, $B_1=B_2\equiv 0$, $A_h(\cdot)=\gamma^{-2}A(1/\overline{F}(\cdot))$ is regularly varying with index $\rho_h=\rho/\gamma$, $\alpha_1=\alpha_2=p-1$, $\beta_1=\beta_2=1$, $\mathbb{E}[Z]=\mathbb{E}[X]$, $\eta^{\ast}=\rho_h=\rho/\gamma$, $A^{\ast}(\cdot)=-(\gamma\rho)^{-1} A(1/\overline{F}(\cdot))$, and 
\[
\Gamma_2 = \frac{1}{\rho} ( (1-\rho) \mathrm{B}((1-\rho)/\gamma-p+1,p) - \mathrm{B}(1/\gamma-p+1,p) ).
\]
It follows that the second-order expansion of $x_{\tau}^{Lp}$ is 
\begin{align*}
    \frac{x_{\tau}^{Lp}}{\Delta_{1}F^{\leftarrow}(\tau)} &= 1+\frac{\gamma(p-1)}{\Delta_{1} F^{\leftarrow}(\tau)} (\mathbb{E}[X]+o(1)) +\gamma\left( \left( \frac{1}{\gamma} \mathrm{B}(1/\gamma-p+1,p) \right)^{-1} - 1 \right) (1-\tau) (1+o(1)) \\
    & + \frac{1}{\rho}\left( ( (1-\rho) \mathrm{B}((1-\rho)/\gamma-p+1,p) - \mathrm{B}(1/\gamma-p+1,p) ) \times \frac{1}{\gamma}\left( \frac{1}{\gamma} \mathrm{B}(1/\gamma-p+1,p) \right)^{\rho-1} \right. \\
    & \left. \qquad \qquad + \left( \frac{1}{\gamma} \mathrm{B}(1/\gamma-p+1,p) \right)^{\rho}-1 +o(1) \right) A((1-\tau)^{-1}) 
\end{align*}
This matches expansion (A14) of~\cite{stuuss2022}, itself a corrected version of Proposition 3 of \cite{daouia2019Lpestimation}: note that, despite the fact that the term in $(1-\tau)$ is not reported as being the same in this Proposition 3, the proof of their Proposition~2 indeed shows that there are two contributions proportional to $(1-\tau)$, due to (with the notation therein) a term called $I_1(q;p)$ in their Equation~(A.10) and the $\overline{F}(q)$ term in their Equation (A.11). In the current setting $\gamma<1/(p-1)<1$, so any term in $(1-\tau)$ is negligible with respect to $1/F^{\leftarrow}(\tau)$. It follows that 
\begin{align*}
    & \frac{x_{\tau}^{Lp}}{\Delta_{1}F^{\leftarrow}(\tau)} \\
    &= 1+ \frac{1}{\rho}\left( ( (1-\rho) \mathrm{B}((1-\rho)/\gamma-p+1,p) - \mathrm{B}(1/\gamma-p+1,p) ) \times \frac{1}{\gamma}\left( \frac{1}{\gamma} \mathrm{B}(1/\gamma-p+1,p) \right)^{\rho-1} \right. \\
    & \left. \qquad \qquad + \left( \frac{1}{\gamma} \mathrm{B}(1/\gamma-p+1,p) \right)^{\rho}-1 +o(1) \right) A((1-\tau)^{-1}) + \frac{\gamma(p-1)}{\Delta_{1} F^{\leftarrow}(\tau)} (\mathbb{E}[X]+o(1)). 
\end{align*}
\end{example}
}

\begin{example}
\label{ex-RDEU1-3-29} (Generalized expectiles) Recall Example 3.4 of \cite{mao2018risk}  {in which} the coherent generalized expectile is defined as the unique solution to the equation
\begin{equation}
\tau \, \mathrm{TVaR}_{p}((X-x)_{+})+(1-\tau)\mathrm{TVaR}_{q}(-(X-x)_{-}%
)=0,\ x\in \mathbb{R},\label{TVaR-Dutch}%
\end{equation}
 {satisfying} $p\leq q$ and $\tau/(1-\tau)\geq (1-p)/(1-q)$. This is an example of the generalized Dutch type II risk measures introduced in~\cite{cai2020risk}. In fact, when $u_{i}(x)=2x$ for $x>0$, and $h_{1}(x)=(x-p)_{+}/(1-p)$ and $h_{2}(x)=(x-q)_{+}/(1-q)$ with $p,q\in\lbrack0,1)$, the generalized shortfall risk measure coincides with the coherent generalized expectile in (\ref{TVaR-Dutch}), denoted by $x_{\tau}^{e}$. Next, we make  {explicit the first- and second-order expansions of $x_{\tau}^{e}$ when $\overline{F} \in 2\mathrm{RV}_{-1/\gamma,\rho/\gamma}$ with $\gamma>0$ and $\rho<0$. For the sake of simplicity, we assume that {\it $F$ is continuous}.}

 {Obviously $u_{i}\in \mathrm{RV}_{1}$ and,} for $x$ large enough, $1-h_{1}( 1-1/x ) = x^{-1}/(1-p)$, so $1-h_{1}( 1-1/\cdot )  \in \mathrm{RV}_{-1}$  {and similarly $1-h_{2}( 1-1/\cdot )  \in \mathrm{RV}_{-1}$. Again the conditions of Theorem \ref{theo-first} reduce to $\gamma<1$ and $\mathbb{E}(|\min(X,0)|)<\infty$, and since $\varphi^{\leftarrow}(t)=U(t/(1-p))$ for large $t$, the first-order expansion of $x_{\tau}^{e}$ reads
\begin{align*}
{x_{\tau}^{e}}=\Delta_{1} ( 1-p )^{-\gamma} F^{\leftarrow}(\tau)(1+o(1)) &= (\gamma^{-1} \mathrm{B}(\gamma^{-1}-1,2))^{\gamma} ( 1-p )^{-\gamma} F^{\leftarrow}(\tau)(1+o(1)) \\
&= (\gamma^{-1}-1)^{-\gamma} ( 1-p )^{-\gamma} F^{\leftarrow}(\tau)(1+o(1)) \ \mbox{ as } \tau\uparrow 1. 
\end{align*}
This can be viewed as an extension of the standard asymptotic equivalent for expectiles in terms of their quantile counterparts.}

 {Then clearly $u_{i}$ are multiples of pure power functions, so $B_1=B_2\equiv 0$, and $1-h_{1}( F(x) ) = \overline{F}(x)/(1-p)$, $1-h_{2}( F(x) ) = \overline{F}(x)/(1-q)$ for $x$ large enough, so $1-h_{1}\left( F(\cdot) \right)
\in2\mathrm{RV}_{-1/\gamma,\rho/\gamma}$ with auxiliary function $\gamma^{-2}A(1/\overline{F}(\cdot))$. Recall also that $\varphi^{\leftarrow}(t)=U(t/(1-p))$ for large $t$ (and then $\varphi^{\leftarrow}((1-\tau)^{-1})=F^{\leftarrow}(p+(1-p)\tau)$). This means that with the notation of Theorem~\ref{main}, $B_1=B_2\equiv 0$, $A_h(\cdot)=\gamma^{-2}A(1/\overline{F}(\cdot))$ is regularly varying with index $\rho_h=\rho/\gamma$, $\alpha_1=\alpha_2=\beta_1=\beta_2=1$, $\mathbb{E}[Z]=\mathbb{E}[X \mathbbm{1}\{ X>F^{\leftarrow}(q) \}]/(1-q)=\mathbb{E}[X|X>F^{\leftarrow}(q)]$, $\eta^{\ast}=\rho_h=\rho/\gamma$, $A^{\ast}(\cdot)=-(\gamma\rho)^{-1} A(1/\overline{F}(\cdot))$, and 
\[
\Gamma_2 = \frac{1}{\rho} ( (1-\rho) \mathrm{B}((1-\rho)/\gamma-1,2) - \mathrm{B}(1/\gamma-1,2) ) = \frac{\gamma^2}{(1-\gamma)(1-\gamma-\rho)}.
\]
Then, after straightforward calculations
\begin{align*}
    & \frac{x_{\tau}^{e}}{\Delta_{1} F^{\leftarrow}(p+(1-p)\tau)} \\
    &= 1+ (1-p)^{\gamma}\frac{\gamma(\gamma^{-1} - 1)^{\gamma}}{F^{\leftarrow}(\tau)} (\mathbb{E}[X|X>F^{\leftarrow}(q)]+o(1)) + \left( (1-\gamma) \frac{1-p}{1-q} - 1\right)(1-\tau)(1+o(1)) \\
    &+ (1-p)^{-\rho}\left( \frac{(\gamma^{-1} - 1)^{-\rho}}{1-\gamma-\rho} + \frac{(\gamma^{-1} - 1)^{-\rho}-1}{\rho} + o(1) \right)  A((1-\tau)^{-1}). 
\end{align*}
Since 
\[
\frac{F^{\leftarrow}(p+(1-p)\tau)}{F^{\leftarrow}(\tau)} = \frac{U((1-\tau)^{-1}/(1-p))}{U((1-\tau)^{-1})} = (1-p)^{-\gamma} \left( 1 + \frac{(1-p)^{-\rho}-1}{\rho} A( (1-\tau)^{-1} )(1+o(1)) \right)
\]
and $\gamma<1$, one may finally conclude that 
\begin{align*}
    & \frac{x_{\tau}^{e}}{\Delta_{1} (1-p)^{-\gamma} F^{\leftarrow}(\tau)} \\
    &= 1+ (1-p)^{\gamma}\frac{\gamma(\gamma^{-1} - 1)^{\gamma}}{F^{\leftarrow}(\tau)} ( \mathbb{E}[X|X>F^{\leftarrow}(q)]+o(1) ) \\
    &+ \left( (1-p)^{-\rho}\left( \frac{(\gamma^{-1} - 1)^{-\rho}}{1-\gamma-\rho} + \frac{(\gamma^{-1} - 1)^{-\rho}-1}{\rho} \right) +\frac{(1-p)^{-\rho}-1}{\rho} + o(1) \right) A((1-\tau)^{-1}). 
\end{align*}
This coincides with the second-order asymptotic expansion of expectiles when $p=q=0$, see Proposition~1 in~\cite{daouia2020expectiles}.}
%{\color{blue}$s=1/\gamma$, $a_{i}=2$, $b=\left(  1-p\right)  ^{-1}$. Then $c^{\ast
%}=(1-p)^{-\gamma}c$ and $c_{0}=(1-p)^{-\gamma}$.
%
%$\Delta_{0}=\frac{\beta_{1}}{\gamma}\mathrm{B}(\beta_{1}/\gamma-\alpha
%_{1},\alpha_{1}+1)=(\gamma^{-1}-1)^{-1}$.
%
%$\Delta_{3}=s^{-1}\Gamma_{2}c_{0}^{\rho_{h}}\Delta_{0}^{\rho_{h}/s-1}%
%=\gamma \Gamma_{2}(1-p)^{-\rho}(\gamma^{-1}-1)^{-\rho+1}$
%
%$\Delta_{5}=s^{-1}c_{0}^{-1}\alpha_{2}\mathbb{E}[Z]\Delta_{0}^{-1/s}%
%=\gamma(1-p)^{\gamma}\mathbb{E}[Z](\gamma^{-1}-1)^{\gamma}$
%
%$\Delta_{6}=s^{-1}c_{0}^{\eta^{\ast}}\Delta_{0}^{\eta^{\ast}/s}=\gamma
%(1-p)^{-\rho}(\gamma^{-1}-1)^{-\rho}$ \bigskip
%
%The term with $A((1-\tau)^{-1})$ has the coefficient
%\begin{align*}
%\Delta_{3}\gamma^{-2}-\Delta_{6}(-(\gamma \rho)^{-1})-\frac{1}{\gamma s\rho}  &
%=\gamma^{-1}\Gamma_{2}(1-p)^{-\rho}(\gamma^{-1}-1)^{-\rho+1}+\rho
%^{-1}(1-p)^{-\rho}(\gamma^{-1}-1)^{-\rho}-\rho^{-1}\\
%& =\frac{\gamma(1-p)^{-\rho}(\gamma^{-1}-1)^{-\rho+1}}{(1-\gamma
%)(1-\gamma-\rho)}+\rho^{-1}(1-p)^{-\rho}(\gamma^{-1}-1)^{-\rho}-\rho^{-1}\\
%& =(1-p)^{-\rho}\left(  \frac{(\gamma^{-1}-1)^{-\rho}}{1-\gamma-\rho}%
%+\frac{(\gamma^{-1}-1)^{-\rho}}{\rho}\right)  -\rho^{-1}%
%\end{align*}
%}
%{\color{orange}
%The term with $\left(  F^{\leftarrow}(\tau)\right)  ^{-1}$ has the coefficient%
%\[
%\Delta_{5}=\gamma(1-p)^{\gamma}\mathbb{E}[X|X>F^{\leftarrow}(q)](\gamma
%^{-1}-1)^{\gamma}.
%\]
%}

 {We examine the accuracy of this expansion when $F$ is the generalized Pareto distribution function $F(x)=1-( \theta/(x+\theta) )  ^{1/\gamma}$ for $x>0$, where $\gamma,\theta>0$. In this setting, $U(t)=\theta (  t^{\gamma}-1)$, and then $U\in2\mathrm{RV}_{\gamma,-\gamma}$ with auxiliary function $A(t)=\gamma t^{-\gamma}$. We take $\gamma=1/3,1/5$ and $\theta=1$, $p=q=0.95$.
In Figure \ref{p1}, by varying $\tau$ from 0.95 to 0.9999, we plot the values obtained through the use of the first- and second-order expansions of ${x_{\tau}^{e}}$. For comparison, the true values of ${x_{\tau}^{e}}$ are also plotted, which are calculated using the
\texttt{uniroot} function in \texttt{R}. From Figure \ref{p1}, it can be seen that the second-order expansion improves the first-order expansion significantly, especially in the lighter-tailed case. 
\begin{figure}[ptb]
	\centering
	\subfloat[$\gamma=1/3$]{\includegraphics[width=0.48\linewidth]{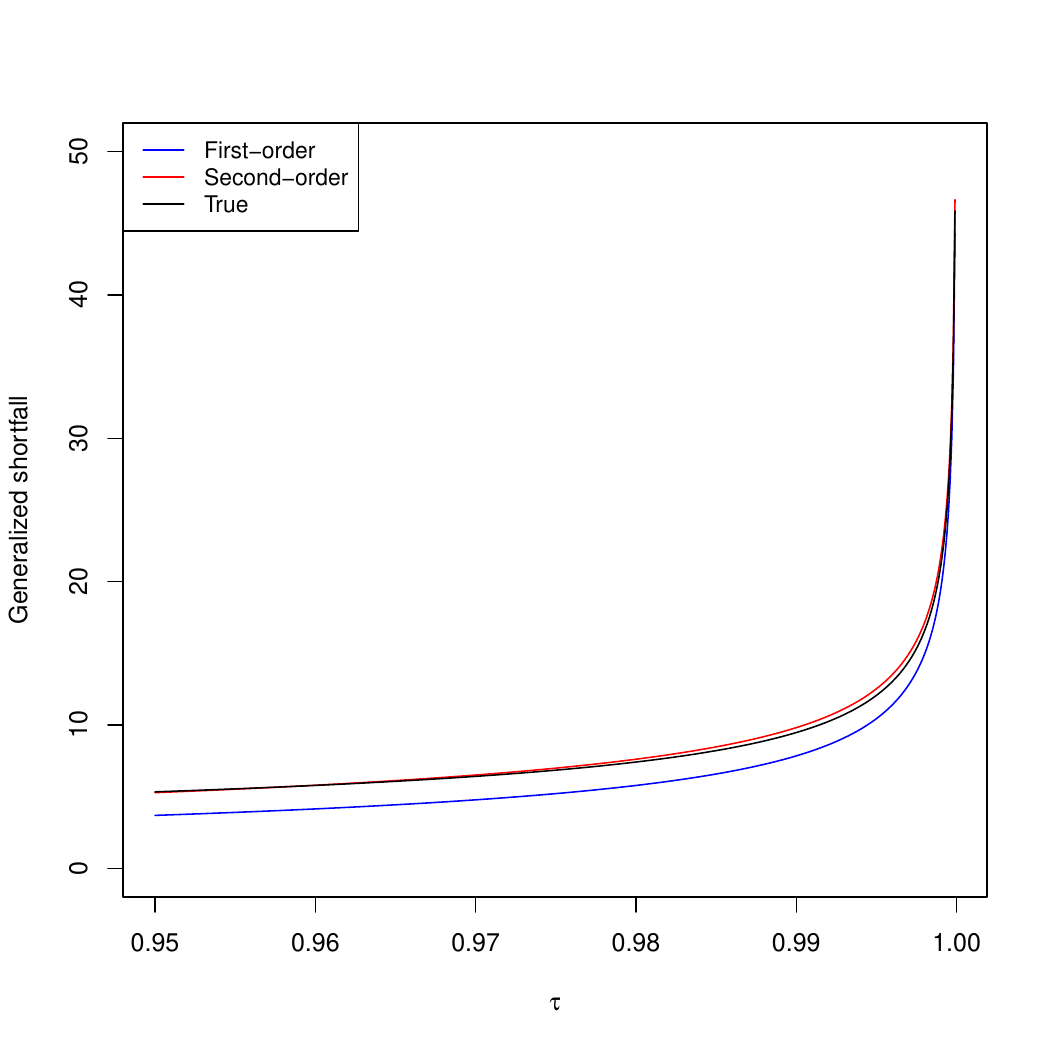} }
	\hfill
	\subfloat[$\gamma=1/5$]{\includegraphics[width=0.48\linewidth]{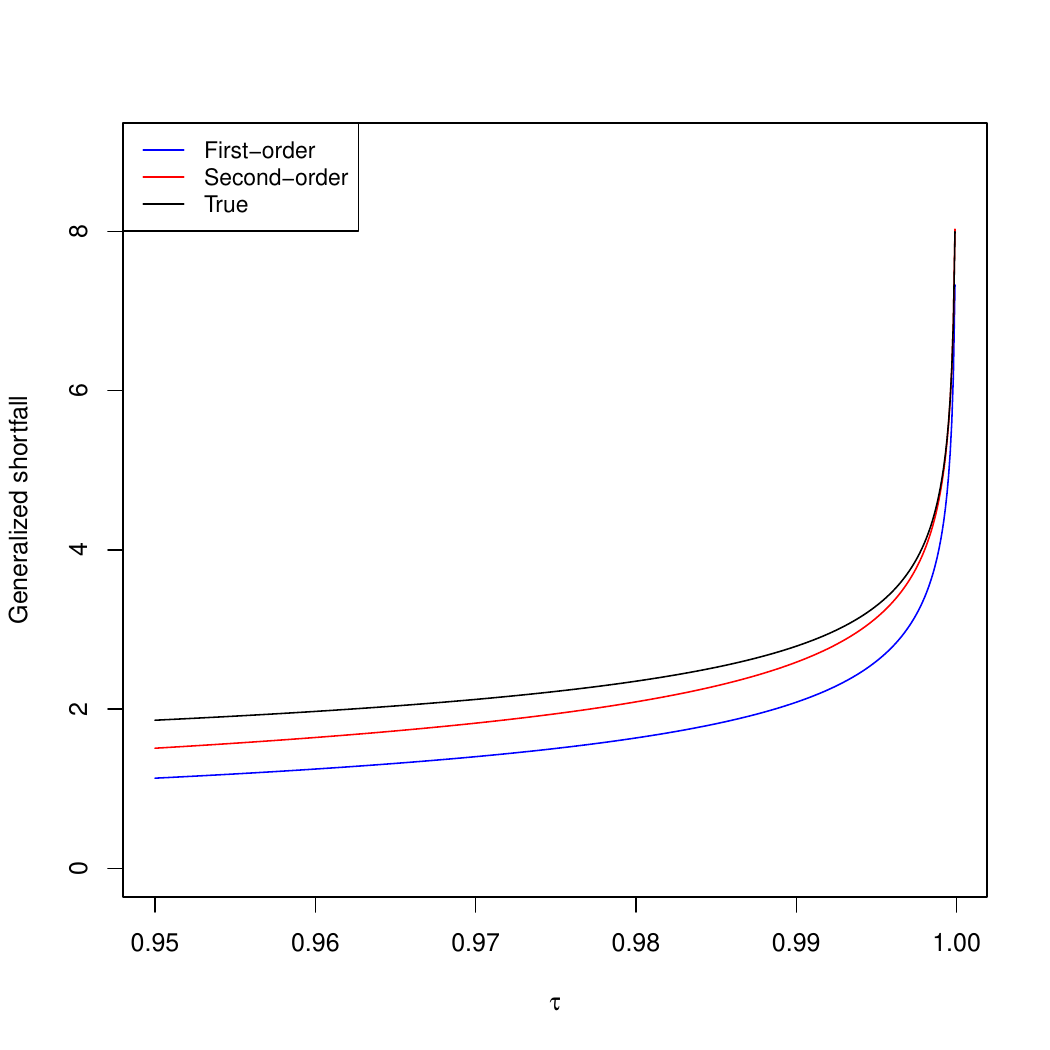}}
	\caption{Comparison of first- and second-order expansions with the true values of the generalized shortfall risk measure of a generalized Pareto distribution.}%
	\label{p1}%
\end{figure}
}
 
\end{example}

 {We now examine the finite-sample performance of the estimator $\widehat{x}_{\tau_n}$ of Section~\ref{sec:estim} in this last example. We consider the following distributions:
\begin{itemize}
\item The pure Pareto distribution with distribution function $F(x)=1-x^{-1/\gamma}$, $x>1$,
\item The Fr\'echet distribution with distribution function $F(x)=\exp(-x^{-1/\gamma})$, $x>0$,
\item The Burr distribution with distribution function $F(x)=1-(1+x^{-\rho/\gamma})^{1/\rho}$, $x>0$ (here $\rho$ is the negative second-order parameter of the distribution). 
\end{itemize}
For each of these three distributions we take $\gamma=1/5$ or $1/3$, and for the Burr distribution we use $\rho=-2$. In each case, we simulate $N=10{,}000$ replications of an independent sample of size $n\in \{500,1{,}000\}$, for which the true generalized expectile risk measure $x_{\tau_n}$ with $p=q=0.95$ and $\tau_n=1-1/n\in \{ 0.998, 0.999 \}$ has been calculated numerically. This was done using the R function {\tt uniroot} in order to find the solution of~\eqref{TVaR-Dutch}, where the function {\tt cubintegrate} from the R package {\tt cubature} has been used beforehand in order to calculate the distorted Tail-Value-at-Risk. In each replication we estimate this risk measure with the estimator introduced in Section~\ref{sec:estim}, where the intermediate level $k=k_n$ is allowed to vary between $n/50$ and $2n/3$ (corresponding respectively to $2\%$ and $66.7\%$ of the total sample size). This produces estimates $\widehat{x}_{\tau_n}^{(j)}(k)$, $j=1,2,\ldots,N$, which are used to calculate the Monte-Carlo approximation to the relative Mean Squared Error (relative MSE) of the estimator $\widehat{x}_{\tau_n}$, that is,
\[
\mathrm{rMSE}(k) = \frac{1}{N} \sum_{j=1}^N \left( \frac{\widehat{x}_{\tau_n}^{(j)}(k)}{x_{\tau_n}} - 1 \right)^2. 
\]
These errors are represented as a function of $k$ in Figures~\ref{fig:n500} and~\ref{fig:n1000} in the twelve situations considered. The MSE tends to be high when $k$ is low, due to the variance of the extreme value estimators dominating in that region, and it also tends to be high when $k$ is large because their bias then dominates, except in the Pareto example for which bias due to the extreme value procedures is exactly 0. Bias is lower in the Burr example than in the Fr\'echet example: this is due to the second-order parameter $\rho=-2$ being further away from 0 in the Burr example than it is in the Fr\'echet example, where it is equal to $-1$. Solving the bias-variance tradeoff produces a stability region for moderately large values of $k$ where MSE is comparatively lower, and this stability region tends to be larger as the second-order parameter gets away from 0. Since the Hill estimator is used in the extrapolation, the asymptotic variance of $\widehat{x}_{\tau_n}$ is asymptotically proportional to $\gamma^2$ by Theorem~\ref{theo-weissman}, so the higher the extreme value index, the higher the MSE should be, just as can be observed by comparing the top and bottom rows in each figure.}

\begin{figure}[h!]
    \centering
    \includegraphics[scale=0.55]{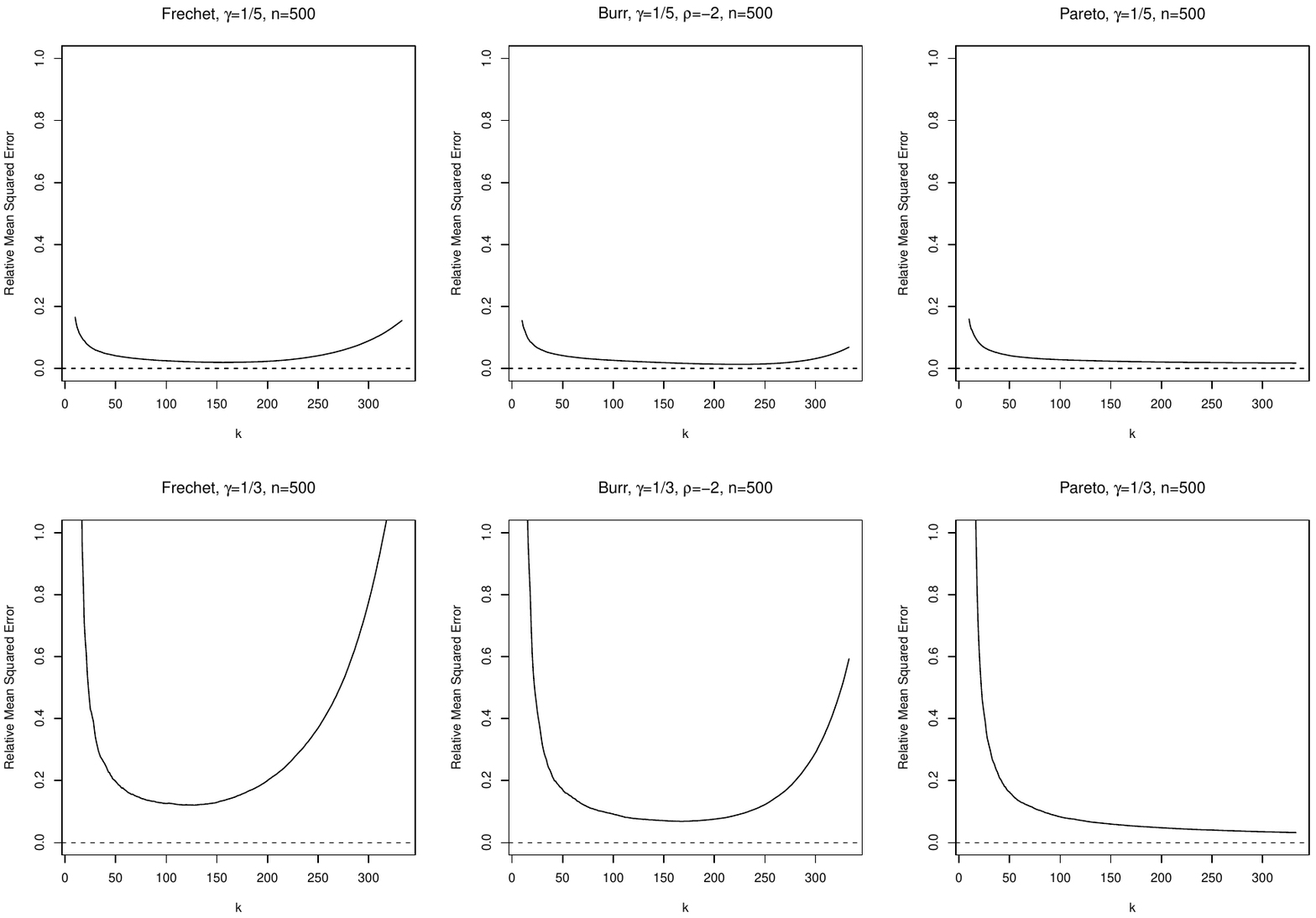}
    \caption{ {Relative Mean Squared Error of the estimator $\widehat{x}_{\tau_n}$, for $n=500$ and $\tau_n=1-1/n=0.998$. Left panels: Fr\'echet distribution, middle panels: Burr distribution with $\rho=-2$, right panels: Pareto distribution. Top panels: $\gamma=1/5$, bottom panels: $\gamma=1/3$.}}
    \label{fig:n500}
\end{figure}

\begin{figure}[h!]
    \centering
    \includegraphics[scale=0.55]{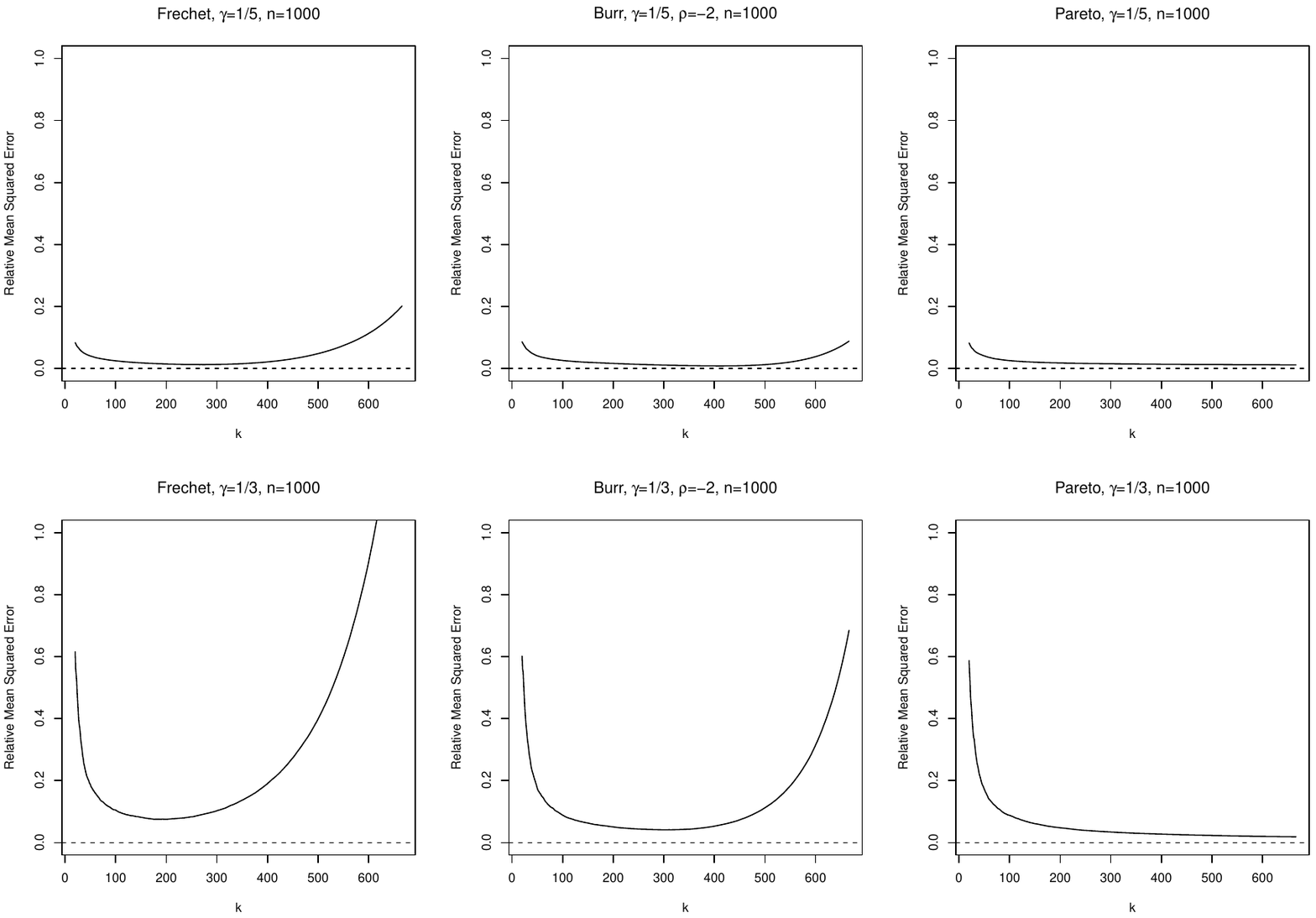}
    \caption{ {Relative Mean Squared Error of the estimator $\widehat{x}_{\tau_n}$, for $n=1000$ and $\tau_n=1-1/n=0.999$. Left panels: Fr\'echet distribution, middle panels: Burr distribution with $\rho=-2$, right panels: Pareto distribution. Top panels: $\gamma=1/5$, bottom panels: $\gamma=1/3$.}}
    \label{fig:n1000}
\end{figure}

\section{Proofs}\label{proof}

\begin{proof}[Proof of Proposition~\ref{prop-xtau}] {Note first that the quantities 
\[
\mathrm{H}_{u_{1},\,h_{1}}((X-x)_{+})=\int_{x}^{\infty}u_{1}(y-x)\mathrm{\,d} h_{1}(F(y)) \mbox{ and } \mathrm{H}_{u_{2},\,h_{2}}((X-x)_{-})=\int_{-\infty}^{x}u_{2}(x-y)\mathrm{\,d}h_{2}(F(y))
\]
are always well-defined, because $u_1$, $u_2$ are positive on $[0,\infty)$ and $h_1\circ F$, $h_2\circ F$ are distribution functions, so that $\mathrm{H}_{u_{1},\,h_{1}}((X-x)_{+})$ and $\mathrm{H}_{u_{2},\,h_{2}}((X-x)_{-})$ are integrals of a positive measurable function with respect to a probability measure. 
\vskip1ex
\noindent
(i) We note that $x\mapsto \mathrm{H}_{u_{1},\,h_{1}}((X-x)_{+})$ and $x\mapsto \mathrm{H}_{u_{2},\,h_{2}}((X-x)_{-})$ are (strictly) decreasing and increasing positive functions on $(x_{\star},x^{\star})$, respectively. Indeed, if $x\in (x_{\star},x^{\star})$ and $\varepsilon>0$ is such that $x+\varepsilon<x^{\star}$, 
\begin{align}
\nonumber
    & \mathrm{H}_{u_{1},\,h_{1}}((X-x)_{+}) - \mathrm{H}_{u_{1},\,h_{1}}((X-(x+\varepsilon))_{+}) \\
\label{eqn:continuityH1}
    &= \int_{x}^{x+\varepsilon} u_{1}(y-x)\mathrm{\,d} h_{1}(F(y)) + \int_{x+\varepsilon}^{\infty} (u_{1}(y-x) - u_{1}(y-x-\varepsilon)) \mathrm{\,d}h_{1}(F(y)) > 0
\end{align}
since $u_1$ is increasing, $u_1(0)=0$ (meaning that the first integral is nonnegative), and $x+\varepsilon<x^{\star}$ (meaning that the second integral is positive). The proof that $x\mapsto \mathrm{H}_{u_{2},\,h_{2}}((X-x)_{-})$ is increasing is similar; the above identity also shows that $x\mapsto \mathrm{H}_{u_{1},\,h_{1}}((X-x)_{+})$ and $x\mapsto \mathrm{H}_{u_{2},\,h_{2}}((X-x)_{-})$ define nonincreasing and nondecreasing functions on $\mathbb{R}$, respectively. Moreover, in the specific case when $x^{\star}<\infty$, one has, for any $x\geq x^{\star}$, 
\[
\mathrm{H}_{u_{1},\,h_{1}}((X-x)_{+})=\int_{x}^{\infty} u_{1}(y-x)\mathrm{\,d} h_{1}(F(y))=0
\]
due to the fact that $F$ is constant equal to 1 on $[x^{\star},\infty)$ and $h_1$ does not have a jump at 1. Similarly $\mathrm{H}_{u_{2},\,h_{2}}((X-x)_{-})=0$ for any $x\leq x_{\star}$ when $x^{\star}>-\infty$. %[This implies that any solution $x_{\tau}$ will necessarily belong to $(x_{\star},x^{\star})$, a fact that will be used in (ii).] 
 {Conclude} by the intermediate value theorem that, since $x\mapsto \mathrm{H}_{u_{1},\,h_{1}}((X-x)_{+})$ and $x\mapsto \mathrm{H}_{u_{2},\,h_{2}}((X-x)_{-})$ are continuous on $(x_{\star},x^{\star})$, the equation 
\[
\tau \, \mathrm{H}_{u_{1},\,h_{1}}((X-x)_{+})-(1-\tau)\, \mathrm{H}_{u_{2},\,h_{2}}((X-x)_{-})=0
\]
has a unique solution which necessarily lies in the interval $(x_{\star},x^{\star})$. 
\vskip1ex
\noindent
(ii) Recall that $x\mapsto \mathrm{H}_{u_{1},\,h_{1}}((X-x)_{+})$ and $x\mapsto \mathrm{H}_{u_{2},\,h_{2}}((X-x)_{-})$ are nonincreasing and nondecreasing, respectively, and that $x_{\tau}<x^{\star}$ for any $\tau<1$. Suppose now that there are $0<\tau<\tau'<1$ such that $x_{\tau}>x_{\tau'}$. Then 
\[
(1-\tau') \mathrm{H}_{u_{2},\,h_{2}}((X-x_{\tau'})_{-}) < (1-\tau) \mathrm{H}_{u_{2},\,h_{2}}((X-x_{\tau})_{-}) = \tau \mathrm{H}_{u_{1},\,h_{1}}((X-x_{\tau})_{+})<\tau' \mathrm{H}_{u_{1},\,h_{1}}((X-x_{\tau'})_{+}).
\]
This is a contradiction because the left- and right-most terms are equal. Hence $x_{\tau}\leq x_{\tau'}$ and $\tau\in (0,1)\mapsto x_{\tau}\in \mathbb{R}$ is nondecreasing, and in particular $x_1=\lim_{\tau\uparrow 1} x_{\tau}$ is well-defined. If $x_1=+\infty$ then obviously $x^{\star}$ is infinite too and $x_1=x^{\star}$; otherwise, we clearly have, for any $\tau\in (0,1)$, 
\[
\tau \mathrm{H}_{u_{1},\,h_{1}}((X-x_1)_{+})\leq \tau \mathrm{H}_{u_{1},\,h_{1}}((X-x_{\tau})_{+}) = (1-\tau) \mathrm{H}_{u_{2},\,h_{2}}((X-x_{\tau})_{-}) \leq (1-\tau) \mathrm{H}_{u_{2},\,h_{2}}((X-x_1)_{-}).
\]
Let $\tau\uparrow 1$ to find $\mathrm{H}_{u_{1},\,h_{1}}((X-x_1)_{+})\leq 0$ and therefore $\mathrm{H}_{u_{1},\,h_{1}}((X-x_1)_{+})=0$. This implies that $x_1\geq x^{\star}$ and then $x_1=x^{\star}$.} 
\vskip1ex
\noindent
{(iii) Clearly $x^{\star}=+\infty$ because $\overline{F}\in \mathrm{RV}_{-1/\gamma}$ with $\gamma>0$. Combine (i) and (ii) to find that it is enough to show the continuity and finiteness of $x\mapsto \mathrm{H}_{u_{1},\,h_{1}}((X-x)_{+})$ and $x\mapsto \mathrm{H}_{u_{2},\,h_{2}}((X-x)_{-})$ on $(x_{\star},+\infty)$. We start by finiteness. Fix $x\in (x_{\star},+\infty)$. Then 
\[
\mathrm{H}_{u_{1},\,h_{1}}((X-x)_{+})=\lim_{T\to +\infty} -\int_{z=0}^{T} u_{1}(z)\mathrm{\,d}(1-h_{1}(1-1/(1/\overline{F}(x+z)))).
\]
Recall that $u_{1}\in \mathrm{RV}_{\alpha_{1}}$ and $u_1$ is bounded on finite intervals of $[0,\infty)$. Then, for any arbitrary $\delta>0$ we have, if $z$ is chosen large enough, that $u_{1}(z)$ is bounded from above by a multiple of $z^{\alpha_1+\delta}$ using Potter bounds (see, e.g. Proposition B.1.9.5~of \cite{de2006extreme}). Moreover, $1-h_{1}(1-1/(1/\overline{F}(x+\cdot)))\in \mathrm{RV}_{-\beta_{1}/\gamma}$. Use the assumption $\beta_{1}/\gamma>\alpha_1$ and Theorem~1.6.5~of \cite{bingolteu1987} to find that $\mathrm{H}_{u_{1},\,h_{1}}((X-x)_{+})$ is indeed finite. The argument for $\mathrm{H}_{u_{2},\,h_{2}}((X-x)_{-})$ is slightly different: write 
\[
\mathrm{H}_{u_{2},\,h_{2}}((X-x)_{-})=\int_{-\infty}^{x/2} u_{2}(x-z) \mathrm{\,d}h_{2}(F(z)) + \int_{x/2}^{x} u_{2}(x-z) \mathrm{\,d}h_{2}(F(z)).
\]
The first term is shown to be finite using Potter bounds and the assumption $\int_{-\infty}^{\infty}|z|^{\alpha_{2}+\delta}\mathrm{\,d}h_{2}(F(z))<\infty$ for some $\delta>0$. The second term is obviously finite because it is bounded from above by $u_{2}(x/2)$. Hence the finiteness of $\mathrm{H}_{u_{2},\,h_{2}}((X-x)_{-})$. We turn to continuity. Recall~\eqref{eqn:continuityH1}: the first term therein clearly converges to 0 as $\varepsilon\downarrow 0$ because $h_1\circ F$ is a distribution function and is therefore right-continuous. The second term, meanwhile, converges to 0 by the dominated convergence theorem, because of the continuity of $u_1$ and $0\leq (u_{1}(y-x) - u_{1}(y-x-\varepsilon)) \mathbbm{1}\{ y\geq x+\varepsilon \} \leq u_{1}(y-x) \mathbbm{1}\{ y\geq x \}$ with the right-hand side being integrable with respect to $\mathrm{d} h_1(F(y))$. 
The continuity of $x\mapsto \mathrm{H}_{u_{2},\,h_{2}}((X-x)_{-})$ is shown similarly.}
\end{proof}

\begin{proof}[Proof of Lemma~\ref{lm-sides}] 
    (i) First note that for any $\varepsilon_{0}>0$,
\begin{align}
\frac{\mathrm{H}_{u_{1},\,h_{1}}((X-x)_{+})}{u_{1}(x)\left(  1-h_{1}%
(F(x))\right)  } &  =\frac{\int_{x}^{\infty}u_{1}(y-x)\mathrm{\,d}h_{1}%
(F(y))}{u_{1}(x)\left(  1-h_{1}(F(x))\right)  }\nonumber \\
&  =\int_{1}^{\infty}\frac{u_{1}(xy-x)}{u_{1}(x)}\mathrm{\,d}\frac
{h_{1}(F(xy))}{1-h_{1}(F(x))}\label{rewrite}\\
&  =\left(  \int_{1+\varepsilon_{0}}^{\infty}+\int_{1}^{1+\varepsilon_{0}%
}\right)  \frac{u_{1}(xy-x)}{u_{1}(x)}\mathrm{\,d}\frac{h_{1}(F(xy))}%
{1-h_{1}(F(x))}\nonumber \\
&  :=  {I_{1}(x)+I_{2}(x)}.\nonumber
\end{align}
Since $u_{1}\in \mathrm{RV}_{\alpha_{1}}$ with $\alpha_{1}>0$ and
$1-h_{1}(1-1/\cdot)\in \mathrm{RV}_{-\beta_{1}}$ with $\beta_{1}>0$, and
$\beta_{1}/\gamma>\alpha_{1}$, using Potter bounds (see, e.g. Proposition
B.1.9.5~of \cite{de2006extreme}), for any $\varepsilon_{1},\delta_{1}>0$, {one
has, for $x$ large enough,}%
\begin{align*}
 {I_{1}(x)} &  \leq \int_{1+\varepsilon_{0}}^{\infty}(1+\varepsilon_{1})\left(
y-1\right)  ^{\alpha_{1}\pm \delta_{1}}\mathrm{\,d}\frac{h_{1}(F(xy))}%
{1-h_{1}(F(x))}\\
&  =(1+\varepsilon_{1})\left(  \varepsilon_{0}^{\alpha_{1}\pm \delta_{1}}%
\frac{1-h_{1}(F(x(1+\varepsilon_{0}))}{1-h_{1}(F(x))}+\int_{1+\varepsilon_{0}%
}^{\infty}\frac{1-h_{1}(F(xy))}{1-h_{1}(F(x))}\mathrm{\,d}\left(  y-1\right)
^{\alpha_{1}\pm \delta_{1}}\right)  .
\end{align*}
 {Now $1-h_{1}(F(\cdot))=1-h_{1}(1-1/(1/\overline{F}(\cdot)))\in \mathrm{RV}_{-\beta_{1}/\gamma}$ and therefore, by Proposition~B.1.10 in~\cite{de2006extreme},
\[
\limsup_{x\to\infty} I_{1}(x)\leq(1+\varepsilon_{1})\left(  \varepsilon_{0}%
^{\alpha_{1}\pm \delta_{1}}\left(  1+\varepsilon_{0}\right)  ^{-\beta
_{1}/\gamma}+\int_{1+\varepsilon_{0}%
}^{\infty}y^{-\beta_{1}/\gamma}\mathrm{\,d}\left(  y-1\right)  ^{\alpha_{1}\pm \delta_{1}} \right)
\]
A similar lower bound applies with $\varepsilon_1$ replaced by $-\varepsilon_1$. Conclude, since $\varepsilon_1$ and $\delta_1$ are arbitrarily small, that}
\[
\lim_{x\rightarrow \infty} {I_{1}(x)}=\varepsilon_{0}^{\alpha_{1}}\left(
1+\varepsilon_{0}\right)  ^{-\beta_{1}/\gamma}+\int_{1+\varepsilon_{0}%
}^{\infty}y^{-\beta_{1}/\gamma}\mathrm{\,d}\left(  y-1\right)  ^{\alpha_{1}}.
\]
Now we turn to $I_{2}$. Since $u_{1}\in \mathrm{RV}_{\alpha_{1}}$ with
$\alpha_{1}>0$, by Proposition B.1.9.6~of \cite{de2006extreme}, there exists
$c>0$  {such that for large enough $x$,
\[
0\leq I_{2}(x)\leq \int_{1}^{1+\varepsilon_{0}}c\mathrm{\,d}\frac{h_{1}(F(xy))}%
{1-h_{1}(F(x))}=c\left(  1-\frac{1-h_{1}(F(x(1+\varepsilon_{0}))}%
{1-h_{1}(F(x))}\right).
\]
[The constant $c$ can be chosen sufficiently large so that it is universal for small values of $\varepsilon_0$, see the proof of Proposition B.1.9.6~of \cite{de2006extreme}.] Hence the bound 
\[
0\leq \liminf_{x\rightarrow \infty}I_{2}(x)\leq \limsup_{x\rightarrow \infty}I_{2}(x) \leq c(1-
(1+\varepsilon_{0})^{-\beta_{1}/\gamma}).
\]
Taking limits as $x\to\infty$ and letting} $\varepsilon_{0}\rightarrow0$, the desired result follows
as
\begin{align}
\lim_{x\rightarrow \infty}\frac{\mathrm{H}_{u_{1},\,h_{1}}((X-x)_{+})}%
{u_{1}(x)\left(  1-h_{1}(F(x))\right)  } &  =\int_{1}^{\infty}y^{-\beta
_{1}/\gamma}\mathrm{\,d}\left(  y-1\right)  ^{\alpha_{1}}\label{equate}\\
&  =\alpha_{1}\mathrm{B}(\beta_{1}/\gamma-\alpha_{1},\alpha_{1}) =  {\frac{\beta_1}{\gamma} \mathrm{B}(\beta_{1}/\gamma-\alpha_{1},\alpha_{1}+1).} \nonumber
\end{align}
 {[Recall the recurrence formula $(x+y)\mathrm{B}(x,y+1)=y \mathrm{B}(x,y)$ valid for any $x,y>0$.]} 

For (ii), let
\[
U_{h_{2}}(x)=\left(  \frac{1}{1-h_{2}(F)}\right)  ^{\leftarrow}(x),\ x>1.
\]
Then $U_{h_{2}}\in \mathrm{RV}_{\gamma/\beta_{2}}$. Again, if $W\sim
\mathrm{Uniform}[0,1]$ then $U_{h_{2}}(1/W)\overset{\mathrm{\,d}}{=}Z\sim
h_{2}(F)$.
%Define $\tilde{u}_{2}(z) = u_2(z)$ on $[0,\infty)$ and $-u_2(-z)$ otherwise.
Then we have
\[
\mathrm{H}_{u_{2},\,h_{2}}((X-x)_{-})=\int_{-\infty}^{x}u_{2}(x-z)\mathrm{\,d}%
h_{2}(F(y))=\mathbb{E}[u_{2}((Z-x)_{-})].
\]
{Consider the split
\[
\frac{\mathbb{E}[u_{2}((Z-x)_{-})]}{u_{2}%
(x)}=\int_{-\infty}^{x/2} \frac{u_{2}(x-z)}{u_{2}(x)}\mathrm{\,d}h_{2}(F(z)) + \int_{x/2}^{x} \frac{u_{2}(x-z)}{u_{2}(x)} \mathrm{\,d}h_{2}(F(z)):=I_{1}(x)%
+I_{2}(x).
\]
To control $I_{1}(x)$, write 
\[
I_1(x)=\int_{\mathbb{R}} \frac{u_{2}(x-z)}{u_{2}(x)} \one\{z\leq x/2\}\mathrm{\,d}h_{2}(F(z))
\]
where we extend the definition of $u_2$ on $\mathbb{R}$ by deciding that $u_2(y)=0$ for $y<0$. Clearly, since $u_2$ is regularly varying, $(u_{2}(x-z)/u_{2}(x)) \one\{z\leq x/2\}\to 1$ pointwise in $z$ as $x\to\infty$. Pick now $\delta>0$ with $\int_{-\infty}^{\infty
}|z|^{\alpha_{2}+\delta}\mathrm{\,d}h_{2}(F(z))<\infty$. Since $u_{2}\in \mathrm{RV}_{\alpha_{2}}$, Potter bounds yield, for $x$ large enough, 
\[
\frac{u_{2}(x-z)}{u_{2}(x)} \one\{z\leq x/2\}\leq(1+\delta)\left(
\frac{x-z}{x}\right)^{\alpha_{2}+\delta} \one\{z\leq x/2\} \leq C_{\delta} (1+|z|)^{\alpha_{2}+\delta}
\]
where $C_{\delta}$ is a positive constant. This is an integrable function with respect to the measure $\mathrm{\,d}h_{2}(F(z))$, so the dominated convergence theorem yields 
\[
\lim_{x\rightarrow \infty}I_{1}(x)=1.
\]
To control $I_{2}(x)$, note that $u_{2}(x)\geq u_{2}(x-z) \geq 0$ when
$x/2\leq z\leq x$ because $u_2$ is increasing. Therefore
\[
0\leq I_{2}(x) \leq \int_{x/2}^{x}\mathrm{\,d}h_{2}(F(z))=h_{2}(F(x))-h_{2}(F(x/2))
\rightarrow 0
\]
as $x\rightarrow \infty$. Thus,
\[
\lim_{x\rightarrow \infty}\frac{\mathbb{E}[u_{2}((Z-x)_{-})]}{u_{2}(x)}=1.
\]
The desired result follows.}
\end{proof}

\begin{proof}[Proof of Theorem~\ref{theo-first}] {The assertions on the regular variation property of $\varphi$ and the existence of the generalized inverse $\varphi^{\leftarrow}$ are immediate, see for example Definition~B.1.8~p.366 of~\cite{de2006extreme}. Combine Equation~\eqref{eq-eqREDU-2} and the first-order expansions in Lemma~\ref{lm-sides} to get
\begin{equation*}
\Delta_{0}(1-h_{1}(F(x_{\tau})))  u_{1}(x_{\tau})\sim (1-\tau)  u_{2}(x_{\tau}).%\label{eqlt}%
\end{equation*}
as $\tau \rightarrow 1$. This is readily seen to be equivalent to $\varphi(x_{\tau}) \sim \Delta_{0} (1-\tau)^{-1}$. When $s>0$, $\varphi^{\leftarrow}\in \mathrm{RV}_{1/s}$, see Proposition~B.1.9.9~p.367 of~\cite{de2006extreme}. It immediately follows, by this same proposition, that
\[
x_{\tau} \sim \varphi^{\leftarrow} ( \varphi(x_{\tau}) ) \sim \varphi^{\leftarrow}(\Delta_{0} (1-\tau)^{-1}) \sim \Delta_{0}^{1/s} \varphi^{\leftarrow}((1-\tau)^{-1})
\]
as $\tau \rightarrow 1$. This is the required result.}
\end{proof}

%\begin{proof}[Proof of Lemma~\ref{lem:phiinvVaR}] 
%Denote $h(\cdot)=1-h_1(1-1/\cdot)\in {\rm RV}_{-\beta}$.
%By definition of $\varphi$ of \eqref{phi}, we have  \begin{align*}
% \frac{\varphi(x)}{h(1/\Fbar(x))} \to c~~{\rm as}~~x\to\infty.
%\end{align*}
%Then by the regularity of $\Fbar$ which implies $h(1/\Fbar(\cdot))\in {\rm RV}_{-\beta\gamma}$, we have
% \begin{align*}
% \frac{\varphi^{-1}(1-\tau)}{U (h^{-1}(1-\tau))} \to c^{1/(\beta\gamma)}~~{\rm as}~~\tau\to1,
%\end{align*}
%that is,
% \begin{align*}
% \frac{\varphi^{-1}(1-\tau)}{{\rm VaR}_{h^{-1}_1(\tau)}(X)} \to c^{1/(\beta\gamma)}~~{\rm as}~~\tau\to1.
%\end{align*}
%This completes the proof.
%\end{proof}

\begin{proof}[Proof of Lemma~\ref{bd-2RV}]
%	(i) 
     {Note} that for any $v>0$,
	\[
	\lim_{t\rightarrow\infty}v^{-\gamma}\frac{\frac{g({vt})}{g(t)}-v^{\gamma}%
	}{B(t)}=	\lim_{t\rightarrow\infty}\frac{(tv)^{-\gamma}g({vt})-t^{-\gamma}g(t)}{t^{-\gamma}%
		g(t)B(t)}=\frac{v^{\rho}-1}{\rho}.
	\]
	This implies that $t^{-\gamma}g(t)\in\mathrm{ERV}_{\rho}$. Since $\rho<0$, by
	Theorem B.2.2 of \cite{de2006extreme}, $g_{0}=\lim_{t\rightarrow\infty}t^{-\gamma}g(t)$
	exists, and $h(t):=g_{0}-t^{-\gamma}g(t)\in\mathrm{RV}_{\rho}$.  {Besides, by} Theorem
	B.2.18 of \cite{de2006extreme}, there exists $\widetilde{B}(t)\sim B(t)$  {(which may be chosen bounded on intervals of the form $(0,t_0]$)} such that
	$t^{-\gamma}g(t)\widetilde{B}(t)=-\rho h(t)$. We have
	\[
	v^{-\gamma}\frac{\frac{g({vt})}{g(t)}-v^{\gamma}}{\widetilde{B}(t)}%
	=\frac{h(t)-h(tv)}{-\rho h(t)}=-\frac{1}{\rho}\left(  1-\frac{h(tv)}%
	{h(t)}\right)  .
	\]
	 {Conclude, by Proposition B.1.9.7 of \cite{de2006extreme}, that for any $\epsilon,\delta>0$, there
	exist $c>0$} and $t_{0}$ such that for all $t\geq t_{0}$  {and $0<v<\delta
	$,}
	\[
	\left\vert \frac{\frac{g({vt})}{g(t)}-v^{\gamma}}{\widetilde{B}(t)}\right\vert
	=-\frac{v^{\gamma}}{\rho}\left\vert \frac{h(tv)}{h(t)}-1\right\vert \leq -\frac{v^{\gamma}}{\rho}\left( 1+\left\vert \frac{h(tv)}{h(t)}\right\vert \right) \leq 
	-\frac{v^{\gamma}}{\rho}(1+c v^{\rho-\epsilon}).
	\]	
	This is the desired result. 
%	(ii) On one hand, from part (i), we have for $0<v+A(t)<\delta$ and $t\geq t_{0}$
%	\begin{align}
%	\left\vert \frac{\frac{g(t(v+A(t)))}{g(t)}-(v+A(t))^{\gamma}}{\widetilde{B}%
%		(t)}\right\vert  &  \leq(v+A(t))^{\gamma}(1+c(v+A(t))^{\rho-\epsilon
%	})\nonumber\\
%	&  \leq(v+1)^{\gamma}(1+c(v-\epsilon)^{\rho-\epsilon}). \label{eq-191026-1}%
%	\end{align}
%	On the other hand,  by  the mean value theorem, for any $\epsilon,\delta>0$, there exists  {$t_0>0$} such that for any $t\geq t_{0}$ and
%	$0<v<1/2$, there exists $\xi\in(v-\epsilon,v+\epsilon)\subset(v-\epsilon,1)$ such that
%	\begin{equation*}
%	 \left\vert \frac{(v+A(t))^{\gamma}-v^{\gamma}}{A(t)}\right\vert =\left\vert
%	\frac{\gamma\xi^{\gamma-1}A(t)}{A(t)}\right\vert \leq\gamma \max\{|v-\epsilon|^{\gamma-1},1\} .
%	\end{equation*}
%by the mean value theorem, we have for $t\geq t_{0}$ and
%	$0<v<1/2$, there exists $\xi\in(v-0.01,v+0.01)$ such that
%	\begin{equation}
%	\left\vert \frac{(v+A(t))^{\gamma}-v^{\gamma}}{A(t)}\right\vert =\left\vert
%	\frac{\gamma\xi^{\gamma-1}A(t)}{A(t)}\right\vert \leq\gamma.
%	\label{eq-191026-2}%
%	\end{equation}
%	This combined with \eqref{eq-191026-1}  yields the result.
\end{proof}

\begin{proof}
[Proof of Lemma~\ref{lm-39-1}] %By Lemma \ref{notation} (ii) and (vi), we have
%\begin{align*}
%1-h_{1}\left(  F(x)\right)   &  =1-h_{1}\left(  1-\overline{F}(x)\right)  \\
%&  =b\left(  c^{1/\gamma}x^{-1/\gamma}\left[  1+\frac{1}{\gamma \rho
%}A(1/\overline{F}(x))+o(A(x))\right]  \right)  ^{\beta_{1}}\left[  1+\frac
%{1}{\varsigma}C(1/\overline{F}(x))(1+o(1))\right]  \\
%&  =bc^{\beta_{1}/\gamma}x^{-\beta_{1}/\gamma}\left[  1+\frac{\beta_{1}%
%}{\gamma \rho}A(1/\overline{F}(x))\left(  1+o(1)\right)  \right]  \left[
%1+\frac{1}{\varsigma}C(1/\overline{F}(x))(1+o(1))\right]  \\
%&  =c_{h}x^{-\beta_{1}/\gamma}\left[  1+\frac{1}{\rho_{h}}A_{h}%
%(x)(1+o(1))\right]  ,
%\end{align*}
%where 
 {Set $c_{h}=bc^{\beta_{1}/\gamma}$. By Lemma \ref{notation} (ii) and (iv), we
have
\begin{align*}
\varphi(x) &  =\frac{u_{2}(x)}{u_{1}(x)(  1-h_{1}(F(x)))  }\\
&  =\frac{a_{2}x^{\alpha_{2}}\left[  1+\frac{1}{\eta_{2}}B_{2}(x)+o(B_{2}%
(x))\right]  }{a_{1}x^{\alpha_{1}}\left[  1+\frac{1}{\eta_{1}}B_{1}%
(x)+o(B_{1}(x))\right] \, c_{h} x^{-\beta_{1}/\gamma}\left[  1+\frac{1}{\rho
_{h}}A_{h}(x)+o(A_{h}(x))\right]  }\\
&  ={\frac{a_{2}}{a_{1}c_{h}}}~x^{s}\left[  1+ \frac{1}{\eta_{2}}%
B_{2}(x)(1+o(1))-\frac{1}{\eta_{1}}B_{1}(x)(1+o(1))-\frac{1}{\rho_{h}}A_{h}(x)(1+o(1))\right]  \\
&  ={\frac{a_{2}}{a_{1}c_{h}}}~x^{s}\left[  1+\left(  \frac{1}{\eta_{2}}%
B_{2}(x)-\frac{1}{\eta_{1}}B_{1}(x)-\frac{1}{\rho_{h}}A_{h}(x)\right)
(1+o(1))\right]  \\
&  =:{\frac{a_{2}}{a_{1}c_{h}}}~x^{s}\left[  1+A^{\ast}(x)(1+o(1))\right]  .
\end{align*}
[In the penultimate line the condition linking $a$, the $b_i$ and $\kappa$ was used to ``merge'' the $o(1)$ terms.] By Lemma \ref{lm-2RV-rep-inv} (ii), we have
\[
\varphi^{\leftarrow}(x)=\left(  {\frac{a_{2}}{c_{h}a_{1}}}\right)
^{-1/s}x^{1/s}\left(  1-\frac{1}{s}A^{\ast}(\varphi^{\leftarrow}%
(x))(1+o(1))\right)  
\]
and $\varphi^{\leftarrow}\in \mathrm{2RV}_{1/s,\eta^{\ast}/s}$, where $\eta^{\ast}=\max \{ \eta_{1},\rho_{h},\eta_{2}\}$. The desired representation of $\varphi^{\leftarrow}(\left(  1-\tau \right)  ^{-1})$ follows.}

 {Then, from the representation of $F^{\leftarrow}(\tau)=U((1-\tau)^{-1})$ in Lemma~\ref{notation}, we have
\begin{align*}
\frac{\varphi^{\leftarrow}(\left(  1-\tau \right)  ^{-1})}{\left(
F^{\leftarrow}(\tau)\right)  ^{1/(\gamma s)}} &  =\frac{c^{\ast}%
(1-\tau)^{-1/s}\left(  1-\frac{1}{s}A^{\ast}(\varphi^{\leftarrow}(\left(
1-\tau \right)  ^{-1}))(1+o(1))\right)  }{\left(  c(1-\tau)^{-\gamma}\left[
1+\frac{1}{\rho}A\left(  \frac{1}{1-\tau}\right)  (1+o(1))\right]  \right)
^{1/(\gamma s)}}\\
&  =\frac{c^{\ast}}{c^{1/(\gamma s)}}\left(  1-\frac{1}{s}A^{\ast}%
(\varphi^{\leftarrow}(\left(  1-\tau \right)  ^{-1}))(1+o(1))-\frac{1}{\gamma s\rho}A\left(  \frac{1}{1-\tau}\right)  (1+o(1))\right).
\end{align*}
It follows that $\varphi^{\leftarrow}(\left(  1-\tau \right)  ^{-1})$ is asymptotically equivalent to $c_{0}(  F^{\leftarrow
}(\tau) )  ^{1/(\gamma s)}$ and then 
\[
\frac{\varphi^{\leftarrow}(\left(  1-\tau \right)  ^{-1})}{\left(
F^{\leftarrow}(\tau)\right)  ^{1/(\gamma s)}}  =c_{0}\left(  1-  \frac{1}{s}A^{\ast}(c_{0}\left(  F^{\leftarrow
}(\tau)\right)  ^{1/(\gamma s)}) (1+o(1)) -\frac{1}{\gamma s\rho}A\left(  \frac
{1}{1-\tau}\right)  (1+o(1))\right)  .
\]
The proof is complete.}
\end{proof}

\begin{proof}[Proof of Lemma~\ref{left2}]
 {Recall (\ref{rewrite}) and write, for $x>0$,}
\begin{align*}
 {\frac{\mathrm{H}_{u_{1},\,h_{1}}((X-x)_{+})}{u_{1}(x)\left(  1-h_{1}%
(F(x))\right)  } - \Delta_0} &  =-\left(  \int_{1}^{\infty}\frac{u_{1}(xy-x)}{u_{1}(x)}%
\mathrm{\,d}\frac{1-h_{1}(F(xy))}{1-h_{1}(F(x))}-\int_{1}^{\infty}\left(
y-1\right)  ^{\alpha_{1}}\mathrm{\,d}y^{-\beta_{1}/\gamma}\right)  \\
&  =-\int_{1}^{\infty} \left( \frac{u_{1}(xy-x)}{u_{1}(x)}-(y-1)
^{\alpha_{1}} \right)\mathrm{\,d}\frac{1-h_{1}(F(xy))}{1-h_{1}(F(x))}\\
&  -\left(  \int_{1}^{\infty}\left(  y-1\right)  ^{\alpha_{1}}\mathrm{\,d}%
\frac{1-h_{1}(F(xy))}{1-h_{1}(F(x))}-\int_{1}^{\infty}\left(  y-1\right)
^{\alpha_{1}}\mathrm{\,d}y^{-\beta_{1}/\gamma}\right)  \\
&  =-\int_{1}^{\infty} \left( \frac{u_{1}(xy-x)}{u_{1}(x)}-(y-1)
^{\alpha_{1}} \right) \mathrm{\,d}\frac{1-h_{1}(F(xy))}{1-h_{1}(F(x))}\\
&  +\int_{1}^{\infty} \left( \frac{1-h_{1}(F(xy))}{1-h_{1}(F(x))}-y^{-\beta_{1}%
/\gamma} \right) \mathrm{\,d}\left(  y-1\right)  ^{\alpha_{1}}\\
&  :=-\int_{1}^{\infty}  {I_{1}(x,y)} \mathrm{\,d}\frac{1-h_{1}(F(xy))}{1-h_{1}%
(F(x))}+\int_{1}^{\infty}  {I_{2}(x,y)} \mathrm{\,d}\left(  y-1\right)  ^{\alpha_{1}}.
\end{align*}
where in the third step we used integration by parts.

We first analyze  {$I_{1}(x,y)$}. Since $u_{1}\in2\mathrm{RV}_{\alpha_{1},\eta_{1}}$
with auxiliary function $B_{1}$,  {there is $\tilde{B}_{1}\sim B_{1}$ such that 
for any (henceforth fixed) $\varepsilon,\delta>0$, there is $x_{0}>0$} such that the following
inequality holds for all $x>x_{0}$ and $xy>x_{0}$,
\begin{equation*}
\left \vert \frac{\frac{u_{1}(x(y-1))}{u_{1}(x)}-(y-1)^{\alpha_{1}}}{\tilde
{B}_{1}(x)}-J_{\alpha_{1},\eta_{1}}(y-1)\right \vert \leq 
\varepsilon(y-1)^{\alpha_{1}+\eta_{1}\pm \delta},%\label{2ndu}%
\end{equation*}
where $y^{\alpha \pm \delta}=y^{\alpha}\max(y^{\delta},y^{-\delta})$  {(and recall that $J_{\gamma,\rho}(x)=x^{\gamma}%
	\frac{x^{\rho}-1}{\rho}$). In particular, if $\varepsilon_0\in (0,1)$ is fixed, then for $x$ large enough, 
\[
\forall y\geq 1+\varepsilon_0, \ J_{\alpha_{1},\eta_{1}}(y-1)-\varepsilon
(y-1)^{\alpha_{1}+\eta_{1}\pm \delta}\leq \frac{I_{1}(x,y)}{\tilde{B}_{1}(x)}\leq J_{\alpha_{1},\eta_{1}}(y-1)+\varepsilon
(y-1)^{\alpha_{1}+\eta_{1}\pm \delta}.
\]
By} integration by parts and since $1-h_{1}(F(\cdot))\in \mathrm{RV}_{-\beta
_{1}/\gamma}$,  {a $\limsup/\liminf$ argument similar to that used in Lemma~\ref{lm-sides} yields,}
\[
 {\lim_{x\rightarrow \infty}\frac{-\int_{1+\varepsilon_{0}}^{\infty}I_{1}(x,y)\mathrm{\,d}\frac{1-h_{1}(F(xy))}{1-h_{1}(F(x))}}{\tilde
{B}_{1}(x)}} = J_{\alpha_{1},\eta_{1}}(\varepsilon_{0})(1+\varepsilon
_{0})^{-\beta_1/\gamma}+\int_{1+\varepsilon_{0}}^{\infty}y^{-\beta_{1}/\gamma
}\mathrm{\,d}J_{\alpha_{1},\eta_{1}}(y-1).
\]
 {Besides, by Lemma \ref{bd-2RV}, there is a constant $C=C(\varepsilon)>0$ such that for $x$ large enough,}
\[
 {1<y<1+\varepsilon_{0} \Rightarrow \left \vert \frac{I_{1}(x,y)}{\tilde{B}_{1}(x)}\right \vert} \leq-\frac
{(y-1)^{\alpha_{1}}}{\eta_{1}}\left(  1+C(y-1)^{\eta_{1}-\varepsilon}\right)
.
\]
Then, as in the proof of Lemma~\ref{lm-sides}, one finds  
\begin{multline*}
     {\limsup_{x\rightarrow \infty}\left \vert \frac{\int_{1}^{1+\varepsilon_{0}}%
I_{1}(x,y)\mathrm{\,d}\frac{1-h_{1}(F(xy))}{1-h_{1}(F(x))}}{\tilde{B}_{1}%
(x)}\right \vert} \\
    \leq -\frac{\varepsilon_{0}^{\alpha_{1}}}{\eta_{1}}\left(
1+C\varepsilon_{0}^{\eta_{1}-\varepsilon}\right)  {(1+\varepsilon
_{0})^{-\beta_1/\gamma}}  -\int_{1}^{1+\varepsilon
_{0}}y^{-\beta_{1}/\gamma}\mathrm{\,d} \left\{ \frac{(y-1)^{\alpha_{1}}}{\eta_{1}%
}\left(  1+C(y-1)^{\eta_{1}-\varepsilon}\right) \right\}  .
\end{multline*}
 {Recall that $\alpha_1+\eta_1>0$, so that the right-hand side above is well-defined and finite for $\varepsilon>0$ small enough, and tends to 0 as $\varepsilon_0\to 0$. Adding up the contributions from $1$ to $1+\varepsilon_0$ and beyond $1+\varepsilon_0$, and letting $\varepsilon_{0}\rightarrow0$, we get
\begin{align*}
\lim_{x\rightarrow \infty}\frac{-\int_{1}^{\infty}I_{1}(x,y)\mathrm{\,d}%
\frac{1-h_{1}(F(xy))}{1-h_{1}(F(x))}}{\tilde{B}_{1}(x)}  &=\int_{1}^{\infty
}y^{-\beta_{1}/\gamma}\mathrm{\,d}J_{\alpha_{1},\eta_{1}}(y-1)\\ 
& = \frac{1}{\eta_1} \int_0^1 u^{\beta_1/\gamma} \left( (\alpha_1+\eta_1) (u^{-1}-1)^{\alpha_1+\eta_1-1} - \alpha_1 (u^{-1}-1)^{\alpha_1-1} \right) \frac{\mathrm{d}u}{u^2} \\ 
& = \frac{1}{\eta_1} \left( (\alpha_1+\eta_1) \mathrm{B}(\beta_1/\gamma-\alpha_1-\eta_1,\alpha_1+\eta_1) - \alpha_1 \mathrm{B}(\beta_1/\gamma-\alpha_1,\alpha_1) \right) \\
& = \frac{\beta_1}{\gamma} \times \frac{1}{\eta_1} ( \mathrm{B}(\beta_1/\gamma-\alpha_1-\eta_1,\alpha_1+\eta_1+1) - \mathrm{B}(\beta_1/\gamma-\alpha_1,\alpha_1+1) ).
%&  =\frac{1}{\eta_{1}}\int_{0}^{1}\left(  v^{-\gamma/\beta_{1}}-1\right)
%^{\alpha_{1}+\eta_{1}}-\left(  v^{-\gamma/\beta_{1}}-1\right)  ^{\alpha_{1}%
%}\mathrm{\,d}v.
\end{align*}
We turn to controlling $I_{2}(x,y)$. By Lemma \ref{notation} (iv), 
\[
\forall y>0, \ \lim_{x\rightarrow \infty}\frac{I_{2}(x,y)}{A_{h}(x)} = J_{-\beta_{1}/\gamma,\rho_{h}}(y)
\]
and for any $\delta>0$, there exist $\tilde{A}_{h}\sim A_{h}$ and $x_{0}>0$ such that for all $x>x_{0}$ and $y\geq 1$,
\[
\left \vert \frac{I_{2}(x,y)}{\tilde{A}_{h}(x)}\right \vert (y-1)^{\alpha_1-1} \leq (y-1)^{\alpha_1-1} (J_{-\beta_{1}/\gamma,\rho_{h}%
}(y) + y^{-\beta_{1}/\gamma+\rho_{h}+\delta}).
\]
If $\delta>0$ is chosen sufficiently small then the right-hand side defines an integrable function on $(1,\infty)$. The dominated convergence theorem then entails}
\begin{align*}
\lim_{x\rightarrow \infty}\frac{\int_{1}^{\infty}  {I_{2}(x,y)} \mathrm{\,d}\left(
y-1\right)  ^{\alpha_{1}}}{\tilde{A}_{h}(x)} &  =\int_{1}^{\infty}%
J_{-\beta_{1}/\gamma,\rho_{h}}(y)\mathrm{\,d}\left(  y-1\right)  ^{\alpha_{1}%
}\\
&  =\frac{\alpha_{1}}{\rho_{h}}\int_{0}^{1}\left(  v^{-\rho_{h}}-1\right)
(1-v)^{\alpha_1-1}v^{\beta_{1}/\gamma-\alpha_1-1}\mathrm{\,d}v \\
&=  {\frac{\alpha_{1}}{\rho_{h}} \left( \mathrm{B}(\beta_1/\gamma-\alpha_1-\rho_h,\alpha_1) - \mathrm{B}(\beta_1/\gamma-\alpha_1,\alpha_1) \right)} \\
&=  {\frac{1}{\rho_h} \left( \left( \frac{\beta_1}{\gamma} - \rho_h \right) \mathrm{B}(\beta_1/\gamma-\alpha_1-\rho_h,\alpha_1+1) - \frac{\beta_1}{\gamma} \mathrm{B}(\beta_1/\gamma-\alpha_1,\alpha_1+1) \right).}
\end{align*}	
 {The proof is complete.}
\end{proof}

\begin{proof}[Proof of Lemma~\ref{Right2}]  {Recall from the proof of Lemma~\ref{lm-sides} that if $Z\sim
h_{2}(F)$, 
\[
\mathrm{H}_{u_{2},\,h_{2}}((X-x)_{-})=\int_{-\infty}^{x}u_{2}(x-z)\mathrm{\,d}%
h_{2}(F(z))=\mathbb{E}[u_{2}((Z-x)_{-})].
\]
Note now} that for any $z<x$,
\begin{equation}
\frac{u_{2}(x)-u_{2}(x-z)  }{u_{2}^{\prime}(x)}=\frac
{u_{2}^{\prime}(\xi)}{u_{2}^{\prime}(x)}z,\label{eq-191005-2}%
\end{equation}
where $\xi \in(x-z,x)$ if $0\leq z<x$ and $\xi \in(x,x-z)$ if $z<0$.
%Also note that
%$$
%\frac{u_{2}\left( w-x\right)-u_{2}(-x) }{u_{2}^{\prime }(-x)} =  \frac{u_{2}^{\prime }(\xi)}{u_{2}^{\prime }(-x)}w,
%$$
%where $\xi\in (-x,-x+w)$ if $w\ge0$ and $\xi\in (w-x,-x)$ if $w<0$.
Also, from \eqref{eq-191005-2},
\[
\lim_{x\rightarrow \infty}\frac{u_{2}(x)-u_{2}(x-z)  }%
{u_{2}^{\prime}(x)}=z
\]
holds for any $z\in \mathbb{R}$,  {because regular variation is locally uniform. Since 
\[
\mathrm{H}_{u_{2},\,h_{2}}((X-x)_{-}) = u_2(x) - u_2(x) (1-h_2(F(x)) - u_2^{\prime}(x)\int_{-\infty}^{x} \frac{u_2(x) - u_{2}(x-z)}{u_2^{\prime}(x)} \mathrm{\,d}%
h_{2}(F(z)), 
\]
we} are left to show that
\begin{equation}
 {\lim_{x\rightarrow \infty}\int_{-\infty}^x \frac{u_{2}(x)-u_{2}(  x-z)  }%
{u_{2}^{\prime}(x)}\mathrm{\,d}h_{2}%
(F(z))=\int_{-\infty}^{\infty}z\mathrm{\,d}h_{2}%
(F(z))},\label{dct}%
\end{equation}
that is, the integral and the limit are interchangeable.  {By Proposition B.1.9.6 of \cite{de2006extreme}, there exist $C>0$}, $x_{0}>0$ such that
for $x\geq x_{0}$, $0<\xi/x\leq1$,
\begin{equation}
\frac{u_{2}^{\prime}(\xi)}{u_{2}^{\prime}(x)}\leq  {C}.\label{ineq-Right2}%
\end{equation}
 {Note that (\ref{ineq-Right2}) holds for the case of $\alpha_{2}=1$ since in this case
by assumption $u_{2}^{\prime}$ is nondecreasing. Moreover, by Proposition B.1.9.5 of \cite{de2006extreme}, for any $\delta>0$, there is $x_1>0$ such that for $x\geq x_1$ and $\xi/x\geq 1$,
\[
\frac{u_{2}^{\prime}(\xi)}{u_{2}^{\prime}(x)} \leq 2\left(
\frac{\xi}{x}\right)^{\alpha_{2}-1+\delta}.
\]
%	{%
%		\begin{align*}
%		&  (1-\delta_{1})\min\left\{  \left(  \frac{-x+w}{-x}\right)  ^{\alpha
%			_{2}-1\pm\delta_{2}},\left(  \frac{-x}{-x}\right)  ^{\alpha_{2}-1\pm\delta
%			_{2}}\right\} \\
%		&  \leq(1-\delta_{1})\left(  \frac{\xi}{-x}\right)  ^{\alpha_{2}-1\pm
%			\delta_{2}}\\
%		&  \leq\frac{u_{2}^{\prime}(\xi)}{u_{2}^{\prime}(-x)}\leq(1+\delta_{1})\left(
%		\frac{\xi}{-x}\right)  ^{\alpha_{2}-1\pm\delta_{2}}\\
%		&  \leq(1+\delta_{1})\max\left\{  \left(  \frac{-x+w}{-x}\right)  ^{\alpha
%			_{2}-1\pm\delta_{2}},\left(  \frac{-x}{-x}\right)  ^{\alpha_{2}-1\pm\delta
%			_{2}}\right\}  ,
%		\end{align*}
Conclude that, with $\xi$ as in~\eqref{eq-191005-2}, that for $x$ large enough, 
\begin{align*}
\forall z<x, \ \left| \frac{u_{2}^{\prime}(\xi)}{u_{2}^{\prime}(x)} z \right| &\leq \left\{ C \mathbbm{1}\{ 0\leq z<x \} + 2 \left(  \frac{x-z}{x}\right)^{\alpha_{2}-1+\delta} \mathbbm{1}\{ z<0 \} \right\} |z| \\
    &\leq \left\{ C + 2 (1-z)^{\alpha_{2}-1+\delta} \mathbbm{1}\{ z<0 \}  \right\} |z|.
\end{align*}
By the assumption that $\int_{-\infty}^{\infty
}|z|^{\alpha_{2}+\delta}\mathrm{\,d}h_{2}(F(z))<\infty$ for some $\delta>0$ 
and the dominated convergence theorem,~(\ref{dct}) holds and therefore, as $x\rightarrow \infty$,
%\[
%\mathbb{E}[u_{2}\left(  Z-x\right)_{-}  ]=u_{2}(x)-u_{2}^{\prime}(x)(\mathbb{E}%
%[Z]+o(1)).
%\]
%Hence the identity
%\[
%\frac{\mathrm{H}_{u_{2},\,h_{2}}((X-x)_{-})}{u_{2}(x)}=1-\frac{u_{2}^{\prime}(x)}{u_{2}(x)}(\mathbb{E}[Z]+o(1)). 
%\]
\[
\frac{\mathrm{H}_{u_{2},\,h_{2}}((X-x)_{-})}{u_{2}(x)}=1-(1-h_2(F(x)))-\frac{u_{2}^{\prime}(x)}{u_{2}(x)}(\mathbb{E}[Z]+o(1)). 
\]
The final identity is obtained by applying Theorem~B.1.5 in~\cite{de2006extreme}.}
\end{proof}

\begin{proof}
[Proof of Theorem~\ref{main} and Theorem~\ref{complete}]  {Combining Equation~\eqref{eq-eqREDU-2} with} Lemmas \ref{left2} and \ref{Right2},
the shortfall risk measure $x_{\tau}$ satisfies 
\begin{align*}
&  \left(  1-h_{1}(F(x_{\tau}))\right)  u_{1}(x_{\tau})\left(  \Delta
_{0}+\Gamma_{1}B_{1}(x_{\tau})(1+o(1))+\Gamma_{2}A_{h}(x_{\tau}%
)(1+o(1))\right) \\
&  =(1-\tau)u_{2}(x_{\tau})\left(  1- {(1-h_{2}(F(x_{\tau})))}-x_{\tau}%
^{-1}(\alpha_{2}\mathbb{E}[Z]+o(1))+(1-\tau)(1+o(1))\right)  ,
\end{align*}
where $Z$ is a random variable having the distribution $h_{2}(F)$. After some
rearrangements, taking $\varphi^{\leftarrow}$ on both sides above yields
\begin{equation}
\varphi^{\leftarrow}\left(  \varphi(x_{\tau})\frac{1- {(1-h_{2}(F(x_{\tau})))}-x_{\tau}^{-1}(\alpha_{2}\mathbb{E}[Z]+o(1))+(1-\tau)(1+o(1))}{\Delta
_{0}+\Gamma_{1}B_{1}(x_{\tau})(1+o(1))+\Gamma_{2}A_{h}(x_{\tau})(1+o(1))}%
\right)  =\varphi^{\leftarrow}(\left(  1-\tau \right)  ^{-1}). \label{lhs}%
\end{equation}
The left-hand side in Equation~\eqref{lhs} can be further rewritten as
\begin{multline*}
\varphi^{\leftarrow}\left(  \Delta_{0}^{-1}\varphi(x_{\tau})\left[  \left(
1-\frac{\Gamma_{1}}{\Delta_{0}}B_{1}(x_{\tau})(1+o(1))-\frac{\Gamma_{2}%
}{\Delta_{0}}A_{h}(x_{\tau})(1+o(1))\right.  \right.  \right. \\
\left.  \left.  \left.  - {(1-h_{2}(F(x_{\tau})))(1+o(1))}-\alpha_{2}\mathbb{E}%
[Z]\frac{1}{x_{\tau}}(1+o(1))+(1-\tau)(1+o(1))\right)  \right]  \right)  .
\end{multline*}
Then by Lemma \ref{lm-2RV-rep-inv} and Lemma \ref{lm-39-1}, with some
calculations we obtain
\begin{align*}
&  \frac{\varphi^{\leftarrow}((1-\tau)^{-1})}{x_{\tau}}=\frac{\varphi
^{\leftarrow}((1-\tau)^{-1})}{\varphi^{\leftarrow}(\varphi(x_{\tau
}))(1+o(A^{\ast}(x_{\tau})))}\\
&  =\Delta_{0}^{-1/s}\left(  1-\left(  \frac{\Gamma_{1}}{s\Delta_{0}}%
B_{1}(x_{\tau})(1+o(1))+\frac{\Gamma_{2}}{s\Delta_{0}}A_{h}(x_{\tau
})(1+o(1))+ {\frac{1}{s} (1-h_{2}(F(x_{\tau}))(1+o(1))} \right.  \right. \\
& \qquad \qquad \qquad \qquad \left.  \left. +\frac{\alpha_{2}\mathbb{E}[Z]}%
{s}\frac{1}{x_{\tau}}(1+o(1)) -\frac{1-{\Delta_{0}^{-\eta^{\ast
}/s}}}{s}A^{\ast}(x_{\tau})(1+o(1))-\frac{1}{s}(1-\tau)(1+o(1))\right)
\right)  .
\end{align*}
Applying Theorem \ref{theo-first}, we have $x_{\tau}\sim \Delta_{0}%
^{1/s}\varphi^{\leftarrow}((1-\tau)^{-1})$, and therefore $B_{1}(x_{\tau}%
)\sim \Delta_{0}^{\eta_{1}/s}B_{1}(\varphi^{\leftarrow}((1-\tau)^{-1})),$
$A_{h}(x_{\tau})\sim \Delta_{0}^{\rho_{h}/s}A_{h}(\varphi^{\leftarrow}%
((1-\tau)^{-1}))$,  {$1-h_{2}(F(x_{\tau}))\sim \Delta_{0}^{-\beta_{2}/(\gamma
s)}(  1-h_{2}(F(\varphi^{\leftarrow}((1-\tau)^{-1}))) )  $} and
$A^{\ast}(x_{\tau})\sim \Delta_{0}^{\eta^{\ast}/s}A^{\ast}(\varphi^{\leftarrow
}((1-\tau)^{-1})).$ The result of Theorem~\ref{main} follows.
Theorem~\ref{complete} is then obtained by applying Lemma~\ref{lm-39-1}.
\end{proof}

\begin{proof}[Proof of Theorem~\ref{theo-weissman}]  {By Lemma~\ref{lm-2RV-rep-inv}(ii), $1/(1-h_1^{-1}(1-1/\cdot))$, the inverse of $1/(h_1(1-1/\cdot))$, is $2\mathrm{RV}_{1,\varsigma}$, and therefore, by Lemma~\ref{lm-2RV-rep-inv}(i), 
\[
\frac{1-\tau_n}{1-h_1^{-1}(\tau_n)} = \frac{1-\tau_n}{1-h_1^{-1}(1-1/(1-\tau_n)^{-1})} \to K\in (0,\infty) \mbox{ as } n\to\infty. 
\]
We then break down $\log(\widehat{x}_{\tau_n}/x_{\tau_n})$ in the following fashion: 
\begin{multline*}
\log\frac{\widehat{x}_{\tau_n}}{x_{\tau_n}} = \log\left( \frac{1}{q_{\tau_n}} \left( \frac{k_n}{n(1-\tau_n)} \right)^{\widehat{\gamma}_n} X_{n-k_n,n} \right) + \log \frac{\Psi(\widehat{\gamma}_n)}{\Psi(\gamma)} + (\widehat{\gamma}_n - \gamma) \log\left( \frac{1-\tau_n}{1-h_1^{-1}(\tau_n)} \right) \\
    + \log\left( \left( \frac{1-\tau_n}{1-h_1^{-1}(\tau_n)}\right)^{\gamma} \frac{q_{\tau_n}}{q_{h_1^{-1}(\tau_n)}} \right) + \log\left( \left( \frac{1}{\gamma} \mathrm{B}(1/\gamma-\alpha,\alpha+1) \right)^{\gamma} \frac{q_{h_1^{-1}(\tau_n)}}{x_{\tau_n}} \right) 
\end{multline*}
with $\Psi(\gamma) = \left( \frac{1}{\gamma} \mathrm{B}(1/\gamma-\alpha,\alpha+1) \right)^{\gamma}$, a continuously differentiable function on the positive half-line. Now, 
\[
\frac{\sqrt{k_n}}{\log(k_n/(n(1-\tau_n)))} \log\left( \frac{1}{q_{\tau_n}} \left( \frac{k_n}{n(1-\tau_n)} \right)^{\widehat{\gamma}_n} X_{n-k_n,n} \right) \stackrel{d}{\longrightarrow} N
\]
by Theorem~4.3.8 p.138 of~\cite{de2006extreme}. It only remains to show that the four other terms in the above decomposition of $\log(\widehat{x}_{\tau_n}/x_{\tau_n})$ are asymptotically negligible. We start by writing  
\[
\frac{\sqrt{k_n}}{\log(k_n/(n(1-\tau_n)))} \log \frac{\Psi(\widehat{\gamma}_n)}{\Psi(\gamma)} = o_{\mathbb{P}}\left( \sqrt{k_n} \log \frac{\Psi(\widehat{\gamma}_n)}{\Psi(\gamma)} \right) = o_{\mathbb{P}}(1)
\]
by the delta-method. Likewise, 
\[
\frac{\sqrt{k_n}}{\log(k_n/(n(1-\tau_n)))} (\widehat{\gamma}_n - \gamma) \log\left( \frac{1-\tau_n}{1-h_1^{-1}(\tau_n)} \right) = O_{\mathbb{P}}\left( \frac{\sqrt{k_n}}{\log(k_n/(n(1-\tau_n)))} (\widehat{\gamma}_n - \gamma) \right) = o_{\mathbb{P}}(1).
\]
The final two terms in the decomposition of $\log(\widehat{x}_{\tau_n}/x_{\tau_n})$ are bias terms. First of all, using assumption $U\in 2\mathrm{RV}_{\gamma,\rho}$, 
\[
\log\left( \left( \frac{1-\tau_n}{1-h_1^{-1}(\tau_n)}\right)^{\gamma} \frac{q_{\tau_n}}{q_{h_1^{-1}(\tau_n)}} \right) = O( A((1-\tau_n)^{-1}) ) = o(A(n/k_n)) 
\]
so that 
\[
\frac{\sqrt{k_n}}{\log(k_n/(n(1-\tau_n)))} \log\left( \left( \frac{1-\tau_n}{1-h_1^{-1}(\tau_n)}\right)^{\gamma} \frac{q_{\tau_n}}{q_{h_1^{-1}(\tau_n)}} \right) = o(1).
\]
Finally, since $q_{h_1^{-1}(\tau_n)}=\varphi^{\leftarrow}((1-\tau_n)^{-1})$, 
\[
\frac{\sqrt{k_n}}{\log(k_n/(n(1-\tau_n)))} \log\left( \left( \frac{1}{\gamma} \mathrm{B}(\alpha,1/\gamma-\alpha+1) \right)^{\gamma} \frac{q_{h_1^{-1}(\tau_n)}}{x_{\tau_n}} \right) = o(1)
\]
by Corollary~\ref{coro:stat} and the assumption $\sqrt{k_n} (k_n/n+|A(n/k_n)|+|B(q_{1-k_n/n})|+|C(n/k_n)|+1/q_{1-k_n/n}) = O(1)$. The proof is complete.}
\end{proof}

\section*{Acknowledgments}

The authors gratefully acknowledge two anonymous referees for their helpful comments which resulted in a substantially improved version of this article, as well as Antoine Usseglio-Carleve for pointing us to the {\tt cubature} package in R for numerical integration tasks. T.~Mao gratefully acknowledges the financial support from Anhui Natural Science Foundation (No.~2208085MA07) and the National Natural Science Foundation of China (No.~71921001). G.~Stupfler gratefully acknowledges support from the French ANR (grant ANR-19-CE40-0013, the EUR CHESS under grant ANR-17-EURE-0010, and the Centre Henri Lebesgue under grant ANR-11-LABX-0020-01), from an AXA Research Fund Award on ``Mitigating risk in the wake of the COVID-19 pandemic'', and from the TSE‐HEC ACPR Chair. F.~Yang gratefully acknowledges financial support from the Natural Sciences and Engineering Research Council of Canada (Grant Number: 04242).

\bibliographystyle{apalike}
\bibliography{reference}

\end{document}